\RequirePackage{fix-cm}
\documentclass[smallcondensed]{svjour3}     % onecolumn (ditto)
\usepackage{amsmath,amsfonts,graphicx,tabularx}
\pdfoutput=1
\usepackage{bm}
\usepackage{float}
\usepackage{newfloat}
\usepackage{multicol}
\usepackage{booktabs}
\usepackage[labelfont=bf]{caption}
\usepackage{scalerel}
\usepackage{lmodern}
\usepackage{nccmath}
\usepackage{mathrsfs}  
\usepackage{mathtools}
\usepackage{enumitem}
\usepackage{tikz}
\usepackage{comment}
\smartqed
\newcommand{\framedbox}[2][0.95\textwidth]{
  \centering
  \tikzstyle{mybox} = [draw=black,line width=1.2pt,inner sep=8pt]
  \begin{tikzpicture}
   \node [mybox] (fig){%
    \begin{minipage}{#1}{#2}\end{minipage}
   };
  \end{tikzpicture}
}
\DeclareMathAlphabet\mathbfcal{OMS}{cmsy}{b}{n}

\DeclareFloatingEnvironment[listname=Box,name=Fig]{story}
\usepackage{algpseudocode}
\algrenewcommand\alglinenumber[1]{\tiny #1:}
\algdef{SE}{Begin}{End}{\textbf{begin}}{\textbf{end}}

\usepackage{etoolbox}

\newcommand{\algstrut}[1][\algruledefaultfactor]{\vrule width 0pt
depth .25\baselineskip height #1\baselineskip\relax}
\newcommand*{\algrule}[1][\algorithmicindent]{\hspace*{.5em}\vrule\algstrut
\hspace*{\dimexpr#1-.5em}}

\makeatletter
\newcount\ALG@printindent@tempcnta
\def\ALG@printindent{%
    \ifnum \theALG@nested>0% is there anything to print
    \ifx\ALG@text\ALG@x@notext% is this an end group without any text?
    \else
    \unskip
    \ALG@printindent@tempcnta=1
    \loop
    \algrule[\csname ALG@ind@\the\ALG@printindent@tempcnta\endcsname]%
    \advance \ALG@printindent@tempcnta 1
    \ifnum \ALG@printindent@tempcnta<\numexpr\theALG@nested+1\relax% can't do <=, so add one to RHS and use < instead
    \repeat
    \fi
    \fi
}%

\patchcmd{\ALG@doentity}{\noindent\hskip\ALG@tlm}{\ALG@printindent}{}{\errmessage{failed to patch}}

\AtBeginEnvironment{algorithmic}{\lineskip0pt}

\newcommand{\norm}[1]{\left\lVert#1\right\rVert}
\newcommand{\abs}[1]{\left\lvert#1\right\rvert}

\makeatletter
\newcommand{\oset}[3][0ex]{%
  \mathrel{\mathop{#3}\limits^{
    \vbox to#1{\kern-2\ex@
    \hbox{$\scriptstyle#2$}\vss}}}}
\makeatother

\AtBeginDocument{%
  \setlength{\oddsidemargin}{\dimexpr(\paperwidth-\textwidth)/2-1in}%
  \setlength{\evensidemargin}{\oddsidemargin}%
  \setlength{\topmargin}{%
    \dimexpr(\paperheight-\textheight)/2-\headheight-\headsep-1in}%
}

\begin{document}

\title{Fully Homomorphic Encryption based on Multivariate Polynomial Evaluation}


\author{\mbox{Uddipana Dowerah$^{1}$  \and Srinivasan Krishnaswamy$^{1}$ }}


\institute{Uddipana Dowerah \at
    %           Department of Electronics and Electrical Engineering\\
				% Indian Institute of Technology Guwahati, Guwahati-781039, India.   \\                   %           Tel.: 0361-258-2544\\
              \email{d.uddipana@iitg.ac.in}           %  \\
           \and
           Srinivasan Krishnaswamy \at
           \email{srinikris@iitg.ac.in} \\~\\
          $^{1}$ Department of Electronics and Electrical Engineering, 
				Indian Institute of Technology Guwahati, Guwahati-781039, India.  % \\                   
              %Tel.: 0361-258-2544\\
}

\date{Received: date / Accepted: date}

\maketitle

\begin{abstract}
We propose a multi-bit leveled fully homomorphic encryption scheme using multivariate polynomial evaluations. The security of the scheme depends on the hardness of the Learning with Errors (LWE) problem. %We first discuss a decision problem called the Hidden Subspace Membership (HSM) problem and its hardness with respect to the well-known Learning with Errors (LWE) problem. %We show that an adversary against the LWE problem can be translated into an adversary against the HSM problem and on the contrary, solving the HSM problem is equivalent to solving the LWE problem with multiple secrets. 
 %We then show that the security of the proposed scheme depends on the hardness of the HSM problem. 
 For homomorphic multiplication, the scheme uses a polynomial based technique that does not require relinearization (and key switching). The noise associated with the ciphertext increases only linearly with every multiplication. %For a depth $L$ circuit, the per gate computation overhead of the scheme is ${\mathcal{O}}(n^3\cdot L^2)$. % where $\lambda$ denotes the security parameter.
\keywords{Fully Homomorphic Encryption \and  Learning with Errors \and Multivariate Polynomials \and Multi-bit Encryption}
\subclass{68P25 \and 94A60}
\end{abstract}

\section{Introduction}
\label{intro}
Homomorphic encryption is the ability to evaluate mathematical functions on encrypted data without decryption. In other words, if $\phi$ denotes a function to be evaluated on the plaintexts $m_1$ and $m_2$, then homomorphic encryption is the ability to compute $\phi(m_1,m_2)$ from their respective encryptions $\textsf{Enc}(m_1)$ and $\textsf{Enc}(m_2)$ without the knowledge of $m_1$ and $m_2$. An encryption scheme is said to be fully homomorphic if functions of arbitrary complexity can be evaluated on the ciphertexts. If an encryption scheme  can evaluate functions of only limited complexity, then it is said to be somewhat homomorphic. %Homomorphic encryption has applications in \cite{}. 

Following the initial construction of a Fully Homomorphic Encryption (FHE) scheme \cite{gentry2009fully}, similar schemes were proposed in \cite{van2010fully,smart2010fully,brakerski2011fully,gentry2011implementing,gentry2012fully,gentry2012better,coron2011fully,cheon2013batch}. In such schemes, a somewhat homomorphic scheme is converted to a fully homomorphic one using \textit{bootstrapping}. %Bootstrapping is the process of `refreshing' a ciphertext by re-encrypting it under a different key and then homomorphically evaluating the decryption circuit on the inner encryption. Hence, the refreshed ciphertext is an encryption of the same message with reduced noise. 
Bootstrapping relies on additional hardness assumptions. As an alternative, a different set of schemes were proposed \cite{brakerski2014efficient,brakerski2014leveled,brakerski2012fully,gentry2013homomorphic} based on the hardness of the Learning with Errors (LWE) problem. In these LWE-based schemes, a (leveled) fully homomorphic encryption scheme can be obtained from a somewhat homomorphic one without bootstrapping. %These schemes can further be classified based on the noise managemant technique used to counter noise in an evaluated ciphertext. 

 The starting point of LWE-based schemes such as \cite{brakerski2014efficient,brakerski2014leveled,brakerski2012fully} is the initial cryptosystem proposed in \cite{regev2005lattices,regev2009lattices}. The ciphertext and the secret key are vectors of size $\mathcal{O}(n)$ and the plaintext is recovered by computing their inner product. Homomorphic multiplication is performed by tensoring two ciphertexts. The corresponding secret key is the tensor product of the original secret key with itself. As a result, homomorphic multiplication squares the size of the ciphertext. This blow up in ciphertext size can be scaled back down to its original size by relinearizing the ciphertext after multiplication. Relinearization is the process of multiplying the ciphertext by an $\tilde{\mathcal{O}}(n^2) \times n$ matrix and the evaluation key must consist of $L$ such matrices in order to evaluate a circuit of depth $L$. 
  In \cite{brakerski2014efficient,brakerski2014leveled}, homomorphic multiplication squares the noise and is countered using a noise management technique called \textit{modulus switching} introduced in \cite{brakerski2014efficient}. In \cite{brakerski2012fully}, the noise grows linearly and a leveled FHE scheme can be obtained without modulus switching.
  
  Another approach to construct an LWE-based leveled FHE scheme without bootstrapping was proposed in \cite{gentry2013homomorphic} using the \textit{approximate eigenvector method}. It avoids the expensive relinearization step and homomorphic addition and multiplication are matrix addition and multiplication respectively. It entirely removes the need for an evaluation key and the noise in evaluated ciphertexts can be controlled using a technique called \textit{flattening}. %However, the per gate computation of the scheme still depends on the multiplicative depth $L$ as in previous LWE-based schemes \cite{brakerski2014leveled,brakerski2012fully}. 
  Further, ring-LWE variants of this scheme were developed in \cite{ducas2015fhew,chillotti2016faster}.

% Multivariate polynomial based schemes called Polly Cracker \cite{fellows1993combinatorial} naturally inherits homomorphic properties. However, Polly Cracker based constructions are vulnerable to linear algebra based attacks and attacks using efficient Gr\"{o}bner basis construction algorithms \cite{levy2009survey}. In an attempt to overcome these limitations, a somewhat homomorphic encryption scheme was proposed in \cite{albrecht2011polly} which is a noisy variant of the Polly Cracker schemes. Another fully homomorphic encryption scheme using multivariate polynomials was proposed in \cite{tamayo2017fully} which was broken immediately in \cite{alperin2017total}.  

%\vspace{-1.5em}
\subsection{Our Contribution}

In this paper, we propose a multi-bit leveled fully homomorphic encryption scheme based on multivariate polynomial evaluations. Here, multiple plaintext bits are encrypted in a single ciphertext. %This leads to a decrease in the ciphertext expansion ratio by a polynomial factor in the security parameter. 
 As in \cite{smart2014fully,brakerski2013packed,peikert2008framework,gentry2012fully}, homomorphic addition and multiplication can be performed simultaneously on multiple plaintext bits. Schemes with this property include the ones described in \cite{smart2014fully,brakerski2013packed,peikert2008framework,gentry2012fully}. 
We introduce the proposed scheme as a symmetric key  scheme and then extend it to a public key variant. The security of the scheme depends on the hardness of the LWE problem.

The proposed scheme is based on the evaluation of multivariate polynomials of a secret ideal $\mathcal{I}$. A ciphertext is obtained by evaluating a random polynomial in $\mathcal{I}$ on a set of distinct points and adding scaled plaintext bits to a number of these evaluations corrupted with noise.  %A random polynomial in $\mathcal{I}$ of degree $\leq r, r \in \mathbb{N}$ is evaluated at a set of $\ell$ distinct points for some $\ell \in \mathbb{N}$. If $n$ denotes the dimension of the subspace $\mathcal{I}_{\leq r}$, then the number of plaintext bits that can be packed into a ciphertext is $\ell-n$. The ciphertext is obtained by adding the scaled plaintext bits to $\ell-n$ of these evaluations corrupted with noise. Hence, an encryption of ($\ell-n$) bits of zeros is a noisy vector of the subspace obtained by evaluating polynomials in $\mathcal{I}_{\leq r}$. %The secret key is a basis for the perpendicular space of this subspace and decryption is performed by computing a matrix-vector product. %the inner product of the ciphertext with the secret key. 
%Multiplication in the scheme is essentially multiplication of the underlying polynomials. It is performed by evaluating a bilinear map on the ciphertexts which is represented by a 3-way tensor. This tensor is given as the public evaluation key for multiplication. 
Multiplication in the scheme is performed by evaluating a bilinear map on the ciphertexts. This map is represented by a 3-way tensor and is given as the public evaluation key for multiplication. The aim here is to use the multiplicative property of the polynomial ring to perform homomorphic multiplication.  %It represents the multiplication of the underlying polynomials. 
 Unlike other LWE-based schemes, homomorphic multiplication does not increase the size of the ciphertexts. Therefore, there is no need for relinearization. The noise associated with the ciphertexts increases only linearly with each homomorphic operation. Hence, a leveled FHE scheme can be obtained without modulus switching. The per gate computation of the scheme is ${\mathcal{O}}(n^3\cdot L^2)$ where $L$ denotes the multiplicative depth of the circuit. 
 
  An attempt to use polynomials for homomorphic multiplication was made in \cite{dowerah2019somewhat}. However, this process required an increase in the ciphertext size with each multiplication.

%A preliminary version of this scheme was presented at the 29th International Conference Radioelektronika 2019 \cite{dowerah2019somewhat}. However, %\cite{dowerah2019somewhat} does not analyze the hardness of the HSM problem. Further, 
%homomorphic multiplication of the scheme in \cite{dowerah2019somewhat} required an increase in the ciphertext size with each multiplication. %In this paper, we give detailed analysis of the hardness of the HSM problem along with 
%This paper presents significant modification of the encryption scheme so as to keep the ciphertext size constant after multiplication. 

\vspace{-1em}
\subsection{Paper Organization}

The remainder of the paper is organized as follows. Section \ref{HSM:sec2} contains the preliminaries for the proposed work. %In Section \ref{HSM:sec3}, we describe the Hidden Subspace Membership problem and discuss its hardness with respect to the Learning with Errors problem. 
In Section \ref{HSM:sec4}, we propose a leveled fully homomorphic encryption scheme based on the hardness of the LWE problem. Finally, in Section \ref{PK}, we discuss the public key variant of the scheme.

\section{Preliminaries} \label{HSM:sec2}
\subsection{Notations}
$\lambda$ denotes the security parameter.
The set of integers and natural numbers are denoted by $\mathbb{Z}$ and $\mathbb{N}$ respectively. $\mathbb{R}$ and $\mathbb{Q}$ denote the set of real and rational numbers respectively. Given a set ${S}$, $x\oset{\scriptscriptstyle{\$}}{\leftarrow}{S}$ means that $x$ is sampled  uniformly at random from ${S}$. The cardinality of $S$ is denoted by $\abs{S}$. %For any positive integer $n$, $[n]$ denotes the set of integers $\{1,2,\hdots,n\}$. 
For a real number $x$, $\lfloor x \rfloor$, $\lceil x \rceil $ and $\lfloor x \rceil$ denote the rounding of $x$ down, up or to the nearest integer. 
For a prime $q$, $\mathbb{Z}_q$ denotes a finite field of cardinality $q$ and its elements are represented by the integers in the interval \(\left[-\left\lfloor \frac{q}{2} \right\rfloor,\left\lfloor \frac{q}{2} \right\rfloor \right] \). %$\mathbb{F}_q[x_1,\hdots,x_l]_{\leq r}$ denotes the set of polynomials in $x_1,\hdots,x_l$ with coefficients in $\mathbb{F}_q$ of degree $\leq r$. 
%Given a set of points $\bm{Z}:=\{\bm{z}_i\in \mathbb{F}_q^l\}_{1 \leq i \leq \ell} $ and a polynomial $f \in \mathbb{F}_q[x_1,\hdots,x_l]_{\leq r}$, $f(\bm{Z})=(f(\bm{z}_1),\hdots,f(\bm{z}_\ell)) \in \mathbb{F}_q^\ell$ denotes the evaluation of $f$ at these points. 
We use `log' to denote logarithm to the base 2. For some $a \in \mathbb{Z}$, we use $(a~{mod}~q)$ and $[a]_q$ interchangeably to denote the modular reduction of $a$ by $q$ into the interval $(-q/2,q/2]\cap \mathbb{Z}$. %Tensors are denoted by uppercase bold script letters $\mathbfcal{A},\mathbfcal{B},\hdots$ etc., matrices are denoted by uppercase bold letters $\bm{A}, \bm{B}, \hdots$ etc. and vectors are denoted by lowercase bold letters $\bm{a}, \bm{b},\hdots$ etc. 
A vector $\bm{v}$ is usually a row vector unless stated otherwise. %The notation $\bm{v}(i)$ is used to denote the $i^{th}$ element of $\bm{v}$. 
We use $(\bm{v},\bm{w})$ to denote the vector $[\bm{v} \mathbin\Vert \bm{w} ]$. The one norm of a vector $\bm{v}$ is denoted by $\norm{\bm{v}}_1$ and the infinity norm by $\norm{\bm{v}}_{\infty}$. The inner product of two vectors $\bm{v}_1, \bm{v}_2$ is denoted using $\langle \bm{v}_1,\bm{v}_2\rangle:=\bm{v}_1\bm{v}_2^T$. The notation $\bm{v}_1 \odot \bm{v}_2$ denotes the component-wise product of $\bm{v}_1$ and $\bm{v}_2$. A function $negl(x): \mathbb{N} \rightarrow \mathbb{R}$ is called negligible if, for every $c \in \mathbb{N}$, there exists an integer $n_c$ such that $\vert negl(x) \vert < \frac{1}{x^c}$ for all $x > n_c$. We write $negl(\cdot)$ to denote an arbitrary negligible function. %For some $a \in \mathbb{Z}$, we use $(a~{mod}~q)$ and $[a]_q$ interchangeably to denote the modular reduction of $a$ by $q$ into the interval $(-q/2,q/2]\cap \mathbb{Z}$.

\subsection{Learning with Errors}
%The LWE problem is to recover a secret $\bm{s} \in \mathbb{Z}_q^n$, given a set of noisy linear equations on $\bm{s}$. 
Learning with Errors is the problem of solving a system of noisy linear equations over $\mathbb{Z}_q$. The LWE problem was introduced in \cite{regev2005lattices} and is an extension of the Learning Parity with Noise (LPN) problem to a higher modulus. It can also be seen as the problem of decoding random linear codes. The LWE problem can be defined as follows.
\begin{definition}
{\bf(Learning With Errors)}.  For positive integers $q,n \in \mathbb{N}$, let $\mathcal{X}$ be a probability distribution on $\mathbb{Z}$ and $\bm{s}$ be a secret vector in $\mathbb{Z}_q^n$, chosen uniformly at random. Given polynomially many samples of the form \((\bm{a}_i,b_i)\in \mathbb{Z}_q^n \times \mathbb{Z}_q \) where $\bm{a}_i $ is sampled uniformly at random from $ \mathbb{Z}_q^n $ and $b_i:=\langle \bm{a}_i,\bm{s}\rangle+e_i~({mod}~q)$ for some $e_i \oset{\scriptscriptstyle{\$}}{\leftarrow} \mathcal{X}$, the search variant of the learning with errors problem, denoted by \textup{LWE}$_{n,q,\mathcal{X}}$, is to output the vector $\bm{s} \in \mathbb{Z}_q^n$ with overwhelming probability.
The decisional variant of the problem denoted by \textup{DLWE}$_{n,q,\mathcal{X}}$ is to distinguish LWE samples from samples chosen uniformly at random from \(\mathbb{Z}_q^n \times \mathbb{Z}_q\). The DLWE problem has been shown to be equivalent to the LWE problem \cite{regev2009lattices}.
\end{definition}

The noise in the LWE problem is usually sampled from a discrete Gaussian distribution.
\begin{definition}
{\bf(Discrete Gaussian distribution)}. For $\alpha>0$ and $q=q(n)$, $\mathcal{X}_{\alpha}$ is a discrete Gaussian distribution on $\mathbb{Z}$ with mean zero and standard deviation $\sigma:=\alpha q$ if for an integer $x $ in $(-q/2,q/2]$
 \begin{align}
 Pr\left[x \oset{\scriptscriptstyle{\$}}{\leftarrow} \mathcal{X}_{\alpha}\right]:= \frac{exp\left(-\pi x^2/\sigma^2\right)}{\sum\limits_{y \in (-q/2,q/2]} exp\left(-y^2/\sigma^2 \right)}
\end{align}  
\end{definition}

A distribution over the integers is said to be bounded when it takes values within a specific interval. %where the magnitude of the samples are bounded by $B$ with very high probability.  
\begin{definition}
{\bf(${B}$-bounded distribution)}. A distribution $\mathcal{X}$ over the set of integers is said to be $B$-bounded if %it takes values in the interval $[-B,B]$. 
\begin{align}
Pr\left[\abs{x}>B ~\vert~ x \oset{\scriptscriptstyle{\$}}{\leftarrow} \mathcal{X}\right]={negl}(n)
\end{align}
\end{definition}

\subsubsection{Hardness of LWE} \label{hardness}
Given that $\mathcal{X}$ in LWE is a discretized Gaussian distribution with $\sigma \geq 2 \sqrt{n}$, statistically 
indistinguishable from a $B$-bounded distribution for some appropriate $B$, there exists a quantum reduction of LWE to approximating the decisional Shortest Vector Problem (GapSVP$_{\gamma}$)\cite{regev2005lattices,regev2009lattices} where $\gamma$ is called the approximation factor and depends on the ratio $q/B$ ($\gamma=(q/B)\cdot \tilde{\mathcal{O}}(n)$). Further, \cite{peikert2009public} and \cite{10.1145/2488608.2488680} gave classical reductions of LWE from the worst-case lattice problems for an exponential modulus and a polynomial modulus respectively.  %there exist quantum and classical reductions of LWE to approximating the decisional shortest vector problem (GapSVP$_{\gamma}$). %Given a lattice of dimension $n$ and a number $d$, GapSVP$_{\gamma}$ is the problem to decide whether the lattice has a vector shorter than $d$ or it does not contain any vector shorter than $\gamma(n) \cdot d$. 
The best-known algorithms for GapSVP$_{\gamma}$ require $2^{\tilde{\Omega}(n/\text{log}\,\gamma)}$ time \cite{schnorr1987hierarchy,micciancio2013deterministic}.  %For $\alpha q > 2 \sqrt{n}$, the decisional variant of the shortest vector problem (G{ap}SVP) reduces to the LWE problem \cite{regev2009lattices,peikert2009public,brakerski2013classical} and the best known algorithms in solving GapSVP runs in exponential time. The LWE$_{n,q,\mathcal{X}}$ problem with the discrete Gaussian distribution with mean 0 and standard deviation $\alpha q$ is considered to be hard when 
%\begin{align*}
%\alpha \geq \frac{3}{2} \cdot \text{max} \left( \frac{1}{q},{2^{-2\sqrt{\ell\,\text{log}\,q\,\text{log}\,\delta}}}  \right)
%\end{align*}
%where $\delta$ is the quality of approximation for the shortest vector problem. For state of the art lattice reduction algorithms, $\delta=1.01$ is hard and $\delta=1.005$ is infeasible \cite{micciancio2011lattice,albrecht2011polly}.
% and is considered to be hard when 
%{ \begin{align}
% \alpha \geq 1.5~\text{max}\left(\frac{1}{q}, 2^{-2\sqrt{n~\text{log}q~\text{log}\delta}}  \right)
% \end{align} }%
% where $\delta$ is the quality of approximation for the shortest vector problem \cite{micciancio2011lattice,regev2010learning}. 
 %There are several search-to-decision reductions of DLWE from the worst case hardness of LWE in the literature \cite{regev2005lattices,peikert2009public,micciancio2012trapdoors}. 
 
 The hardness of LWE can be summarized in terms of the following theorem which summarizes the results in  \cite{regev2009lattices},\cite{peikert2009public} and \cite{10.1145/2488608.2488680}. We state the results in terms of the bound $B$ as stated in \cite{brakerski2012fully,gentry2013homomorphic}.
\begin{theorem}
\textup{({\cite{regev2009lattices,peikert2009public,10.1145/2488608.2488680}}).} Let $q=q(n)$ be prime and let $B \geq \omega(\textup{log}\,n) \cdot \sqrt{n}$. Then, there exists an efficiently sampleable $B$-bounded distribution $\mathcal{X}$ such that if there is an efficient algorithm that solves the average-case DLWE$_{n,q,\mathcal{X}}$ problem, then % . Given samples from the distribution of LWE$_{n,q,\mathcal{X}}$, if there exists an algorithm that solves the LWE problem for parameters $n,q,\mathcal{X}$, then 
\begin{itemize}
\item There exists an efficient quantum algorithm for solving GapSVP$_{\tilde{\mathcal{O}}(n\cdot q/B)}$ on any $n$-dimensional lattice \cite{regev2009lattices}.

\item If $q \geq 2^{n/2}$, then there exists an efficient classical algorithm that solves the GapSVP$_{\tilde{\mathcal{O}}(n\cdot q/B)}$ problem on any $n$-dimensional lattice \cite{peikert2009public}.

\item There exists an efficient classical algorithm that solves the GapSVP$_{\tilde{\mathcal{O}}(n\cdot q/B)}$ problem on any lattice of dimension $\sqrt{n}$ \cite{10.1145/2488608.2488680}.
\end{itemize}
\end{theorem}

\subsection{Tensors}

%A tensor can be defined as a multilinear map $\mathbfcal{T}: \mathcal{V}_1 \times \mathcal{V}_2 \times \cdots \times \mathcal{V}_n \rightarrow \mathcal{W}$, where $\mathcal{V}_1,\hdots,\mathcal{V}_n$ and $\mathcal{W}$ are finite-dimensional vector spaces. Given fixed bases for the vector spaces, a tensor can be represented by a multidimensional array. The order of a tensor is the number of indices required to represent a component of the array. 

  Let $\phi_{\mathbfcal{T}}: \mathcal{V}_1 \times \mathcal{V}_2 \times \cdots \times \mathcal{V}_n \rightarrow \mathcal{W}$ be a multilinear map, where $\mathcal{V}_1,\hdots,\mathcal{V}_n$ and $\mathcal{W}$ are finite-dimensional vector spaces. Given fixed bases for the vector spaces, $\phi_{\mathbfcal{T}}$ can be represented by a multidimensional array $\mathbfcal{T}$ called a tensor. The order of a tensor is the number of indices required to represent a component of the array. 

\textit{Slices} in a tensor are two-dimensional sections generated by fixing all indices except two. In a third-order tensor, slices generated by keeping the last index fixed are called frontal slices. These are analogous to matrices. Therefore a third-order tensor $\mathbfcal{T}^{I_1 \times I_2 \times I_3}$ is an $I_3$ array of $I_1 \times I_2$ matrices. %An order-3 tensor $\mathcal{A}^{I_1 \times I_2 \times I_3 }$ is an $I_3$ array of $ {I_1 \times I_2 } $ order-2 tensors. 

A \textit{Bilinear map} is a function $\phi:\mathcal{V}_1 \times \mathcal{V}_2 \rightarrow \mathcal{V}_3$ that takes two elements from two vector spaces $\mathcal{V}_1$ and $\mathcal{V}_2$ and maps it to an element of a third vector space $\mathcal{V}_3$ such that it is linear in each of its elements. %i.e., for a fixed $\bm{v}_1 \in \mathcal{V}_1$, $\bm{v}_1 \mapsto \phi(\bm{v}_1,\bm{v}_2)$ is a linear function from $\mathcal{V}_1$ to $\mathcal{V}_3$ and for a fixed $\bm{v}_2 \in \mathcal{V}_2$, $\bm{v}_2 \mapsto \phi(\bm{v}_1,\bm{v}_2)$ is a linear function from $\mathcal{V}_2$ to $\mathcal{V}_3$. 
An order-3 tensor $\mathbfcal{T}$ can be used to represent a bilinear map $\phi:\mathcal{V} \times \mathcal{V} \rightarrow \mathcal{V}$, on the vector space $\mathcal{V}$. If $dim(\mathcal{V})=n$ and $\{\bm{b}_1,\hdots,\bm{b}_n \}$ denotes a basis for $\mathcal{V}$, then 
\begin{align}
\phi(\bm{b}_i,\bm{b}_j)= \sum_{i=1}^{n}\sum_{j=1}^{n}\sum_{k=1}^{n} T_{ijk} \cdot \bm{b}_k
\end{align}
For a fixed $k$, $\bm{T}_k:=(T_{ijk})_{{{1 \leq i,j \leq n}}}$ represents a unique matrix of order $n \times n$. For $1 \leq k \leq n$, $\bm{T}_k$ forms the frontal slices of the tensor $\mathbfcal{T}$. The bilinear map $\phi$ then acts on two arbitrary vectors $\bm{v}_1=[{v}_{1,1} ~ v_{2,1} ~ \cdots ~ v_{n,1}],\bm{v}_2=[v_{1,2} ~ v_{2,2}~\cdots ~v_{n,2}] \in \mathcal{V}$ as 
 \begin{align} \label{bilinear map}
\phi(\bm{v}_1,\bm{v}_2) %= \phi\left(\sum_{i=1}^{n} v_{i,1}\bm{b}_i, \sum_{j=1}^{n} v_{j,2}\bm{b}_j \right) 
			&= \sum_{i,j=1}^{n} v_{i,1} \phi(\bm{b}_i,\bm{b}_j )v_{j,2} 
			= \begin{bmatrix}
		\bm{v}_1 \bm{T}_{1} \bm{v}_2^T  & \hdots & \bm{v}_1 \bm{T}_{n} \bm{v}_2^T
		\end{bmatrix}   %~~~~~\square
\end{align} 

\subsubsection{The $n$-mode Product}
The $n$-mode product defines multiplication of a tensor by a matrix. In general, the elementwise $n$-mode product of a tensor $\mathbfcal{T}^{I_1 \times I_2 \times \cdots \times I_N}$ and a matrix $\bm{M}^{J \times I_n}$ is defined as:
\begin{eqnarray}
{(\mathbfcal{T}\times \bm{M})}_{i_1 i_2 \cdots i_{n-1} j i_{n+1} \cdots i_N} &=& \sum_{i_n=1}^{I_n} \mathbfcal{T}_{i_1 i_2 \cdots i_{n-1} i_n i_{n+1} \cdots i_N} \bm{M}_{j,i_n}
\end{eqnarray}
The resultant tensor is of the order of $(I_1 \times I_2 \times \cdots \times I_{n-1} \times J \times I_{n+1} \times \cdots \times I_N)$.

\section{The Proposed Leveled FHE Scheme} \label{HSM:sec4}

In this section, we discuss the construction of the leveled FHE scheme.  For simplicity, we present a private key variant of this scheme. Subsequently, we explain its conversion to a public key scheme. %Although the key size in the private key scheme is large, the key size in the public key variant is comparable to other LWE based schemes. 
%Finally, we introduce a batch variant which can simultaneously encrypt multiple bits.

%\vspace{-1.5ex}

Let us consider an ideal $\mathcal{I} $ of the polynomial ring $ \mathbb{Z}_q[x_1,\hdots,x_v]$, where $q=q(\lambda)$ is prime. For some $r \in \mathbb{N}$, $ \mathbb{Z}_q[x_1,\hdots,x_v]_{\leq r} $ is a vector space of dimension $N:=\binom{v+r}{r}$ over $\mathbb{Z}_q$. Let $\mathcal{I}_{\leq r}$ denote the set of polynomials in $\mathcal{I}$ with degree $\leq r$. Then, $\mathcal{I}_{\leq r}$ is a subspace of $ \mathbb{Z}_q[x_1,\hdots,x_v]_{\leq r} $. Let $ n $ be the dimension of $\mathcal{I}_{\leq r}$. It is very easy to see that a polynomial $f \in \mathbb{Z}_q[x_1,\hdots,x_v]_{\leq r}$ evaluated at all points of $ \mathbb{Z}_q^{v} $ generates a vector in $ \mathbb{Z}_q^{\scaleto{q^v}{6.8pt}} $. The set of such vectors obtained by evaluating all polynomials in $ \mathbb{Z}_q[x_1,\hdots,x_v]_{\leq r} $ constitutes an $N$-dimensional subspace of $ \mathbb{Z}_q^{\scaleto{q^v}{7pt}} $. Similarly, evaluating polynomials in $\mathcal{I}_{\leq r}$ gives us an $ n $-dimensional subspace of $ \mathbb{Z}_q^{\scaleto{q^v}{7pt}} $. Let $ \{\bm{z}_1,\hdots,\bm{z}_\ell\} $ be $\ell$ distinct points in $ \mathbb{Z}_q^{v} $ for some $\ell \in \mathbb{N}$ such that $n < \ell \leq N$. We can always choose the $\ell$ points in such a way that at least $n$ of them are linearly independent. Then, evaluating polynomials in $\mathcal{I}_{\leq r}$ at $ (\bm{z}_1,\hdots,\bm{z}_\ell) $ gives us an $ n $-dimensional subspace $\mathcal{S}_{\scriptscriptstyle \mathcal{I}_{\leq r}}$ of $\mathbb{Z}_q^\ell$. It can be summarized in terms of the following lemma.
\begin{lemma}
For $n, \ell \in \mathbb{N}$ with $n < \ell$, we can find $\ell$ distinct points $\{\bm{z}_i \in \mathbb{Z}_q^v\}_{1 \leq i \leq \ell}$ such that evaluating polynomials in $\mathcal{I}_{\leq r}\subseteq \mathbb{Z}_q[x_1,\hdots,x_v]_{\leq r}$ with $dim(\mathcal{I}_{\leq r}):=n$ at $(\bm{z}_1,\hdots,\bm{z}_\ell)$ spans an $n$-dimensional subspace of the vector space $\mathbb{Z}_q^\ell$.
\end{lemma}
 
The set of evaluation points $ \{\bm{z}_1,\hdots,\bm{z}_\ell\} \in \mathbb{Z}_q^{v} $ must satisfy the following conditions:
\begin{enumerate}
\item Every vector in  $ \mathbb{Z}_q^\ell $ can be got by evaluating a polynomial in $ \mathbb{Z}_q[x_1,\hdots,x_v]_{\leq r} $ at $ (\bm{z}_1,\hdots,\bm{z}_\ell )$.
\item Every vector in $\mathbb{Z}_q^{n}$ can be got by evaluating a polynomial in $\mathcal{I}_{\leq r}$ at $ (\bm{z}_1,\hdots,\bm{z}_{n} )$.
\end{enumerate} 

Because of the conditions imposed on $(\bm{z}_1,\hdots,\bm{z}_\ell)$, we can find a basis $\{\tilde{\bm{s}}_{n+1},\tilde{\bm{s}}_{n+2},\hdots,\tilde{\bm{s}}_\ell\} \in \mathbb{Z}_q^\ell$ for $(\mathcal{S}_{\scriptscriptstyle \mathcal{I}_{\leq r}})^{\perp}$ which has the following form:
\begin{align} \label{sk}
\begin{rcases}
\tilde{\bm{s}}_{n+1} &=\begin{bmatrix} s_{n+1}^1 & \hdots & s_{n+1}^n & 1 & 0 & \hdots & 0 \end{bmatrix}\\
 \tilde{\bm{s}}_{n+2} &=\begin{bmatrix} s_{n+2}^1 & \hdots & s_{n+2}^n & 0 & 1 & \hdots & 0\end{bmatrix} \\
 &~\vdots \\
  \tilde{\bm{s}}_\ell &=\begin{bmatrix} s_{\ell}^1 & \hdots & s_{\ell}^n & 0 & 0 & \hdots & 1\end{bmatrix}
  \end{rcases}
  \end{align}
  
  Alternatively, $\mathcal{S}_{\scriptscriptstyle \mathcal{I}_{\leq r}}$ is the null space of the matrix $\begin{bmatrix}
 \bm{S} & \bm{I}_{\ell-n} \end{bmatrix}$  where $\bm{S}(i,j)=s_{n+i}^j$ for $1 \leq i \leq \ell-n$ and $1 \leq j \leq n$.
In this work, we assume that the choice of $\mathcal{I}$ and the points $\bm{z}_1,\bm{z}_2,\ldots,\bm{z}_\ell$ is such that the rows of the matrix $\bm{S}$ are linearly independent.

\subsection{The Basic Encryption Scheme}

We describe the basic encryption scheme in terms of the following algorithms. (The homomorphic operations are described separately in Section \ref{HP}.) The plaintext space is $\{0,1\}^{\ell-n}$ and the ciphertexts are vectors in $\mathbb{Z}_q^\ell$. Further, we restrict ourselves to the case where $\mathcal{I} $ is a principal ideal. %As we will later see, for suitable parameters, this is good enough to ensure that the distribution of the secret key is uniform. 
%The scheme consists of the following algorithms.  %All operations are performed over $\mathbb{Z}$ and for some $a \in \mathbb{Z}$, $(a~\text{mod}~q)$ denotes an integer in the interval $(-q/2,q/2]$. 
%\vspace{0.5ex}
  \begin{itemize}[leftmargin=0pt, itemindent=20pt,labelwidth=15pt, labelsep=5pt, listparindent=0.7cm,align=left]
\setlength\itemsep{1.5ex}
\item[$ \bullet$]\textsf{Setup}($1^{\lambda}, 1^L$): Takes as input the security parameter $\lambda$ and a parameter $L$ and 
outputs $n=n(\lambda,L)$, $\ell=\ell(\lambda,L)$, modulus $q=q(\lambda,L)$ and noise distribution $\mathcal{X}=\mathcal{X}(\lambda,L)$ such that $\abs{\mathcal{X}} \leq B$. Here $L$ denotes the depth of the circuit that can be homomorphically evaluated. Let $\pi=(n,\ell,q,\mathcal{X})$.

\item[$ \bullet$] \textsf{KeyGen}($\pi$): Choose two positive integers $v$ and $r^{\prime}$ such that $n =  \binom{v+r^{\prime}}{r^{\prime}}$. Choose integers $r$ and $r_g$ such that $r-r_g = r^{\prime}$. %Further, $r_g$ is chosen such that $\binom{v+r_g-1}{r_g-1} $ is greater than $\ell$.
 Choose $\ell$ points, $\bm{z}_1,\ldots,\bm{z}_{\ell}$, from $\mathbb{Z}_q^{v}$ such that the aforementioned conditions are satisfied. %$\binom{v+r_g-1}{r_g-1} $ (dimension of $RM(r^{\prime},v)$) being greater than $\ell$ ensures the existence of such a set of points.
 Choose a random polynomial $g(x_1,\ldots,x_v)$ with degree $r_g$ in $v$ variables such that $g(\bm{z}_1),\hdots,g(\bm{z}_{\ell})$ are all non zero. This polynomial acts as the generator of the ideal $\mathcal{I}$. Generate a basis $\mathcal{B}_{\scriptscriptstyle \mathcal{I}_{\leq r}}=(gh_1,\hdots,gh_n)$ for the subspace $\mathcal{I}_{\leq r}$ by considering linearly independent polynomials $h_1,\ldots,h_n$ having degree less than or equal to $r^{\prime}$.  A basis for $\mathcal{S}_{\scriptscriptstyle \mathcal{I}_{\leq r}}$ is obtained by evaluating the polynomials in $\mathcal{B}_{\scriptscriptstyle \mathcal{I}_{\leq r}}$ at the points $(\bm{z}_1,\hdots,\bm{z}_\ell)$. Construct  a basis of $(\mathcal{S}_{\scriptscriptstyle \mathcal{I}_{\leq r}})^{\perp}$ (as given in Equation \ref{sk}) which can be written in terms of the matrix $\begin{bmatrix}
 \bm{S} & \bm{I}_{\ell-n}
 \end{bmatrix}$ where
  \vspace{-1em}
\begin{align}
    \bm{S}=\begin{bmatrix}
    s_{n+1}^{1} & s_{n+1}^2 & \hdots & s_{n+1}^n \\
    s_{n+2}^1 & s_{n+2}^2 & \hdots & s_{n+2}^n \\
    \vdots & \vdots & \ddots & \vdots \\
    s_{\ell}^1 & s_\ell^2 & \hdots & s_\ell^n
    \end{bmatrix}
\end{align}
%Choose any one vector $\bm{s}_{{\scriptscriptstyle \text{HSM}}}$ from the set $\bm{B}_{\scriptscriptstyle \mathcal{I}_{\leq r}}^{\perp}$.
%Set $\bm{s}_{{\scriptscriptstyle \text{HSM}}}=\tilde{\bm{s}}_j $ for any $n+1 \leq j \leq \ell$. 
%where $s_j^i$ for $1 \leq i \leq n$ and $n+1 \leq j \leq \ell$ denotes the $i^{th}$ entry of the vector $\tilde{\bm{s}}_j$ from Equation (\ref{sk}). 
Choose a matrix ${\bm{R}} \in \mathbb{Z}_q^{\ell \times \ell}$ such that ${\bm{R}}$ is of the form
\vspace{1ex}
\begin{align} \label{R}
{\bm{R}}:=\left[
\begin{array}{c|c}
\bm{R}_1^{n \times n} &  {0}^{n \times(\ell-n)}\\
\midrule
\bm{R}_2^{(\ell-n)\times n} & \bm{I}_{\ell-n}
\end{array}
\right]
\end{align}
where $\bm{R}_1 \in \mathbb{Z}_q^{n \times n}$ is a randomly chosen full rank matrix and the entries of $\bm{R}_2$ are chosen uniformly at random from $\mathbb{Z}_q$. % in such a way that $\ell_1(\bm{s}_i \bm{R}^{-1}) \leq B$ for all $i$ 
 $\bm{I}_{\ell-n}$ denotes the identity matrix of size $\ell-n$.  The secret key is $sk=(\bm{S},\bm{R}_1,\bm{R}_2)$. The public parameters are $n,\ell$ and $q$.

\item[$  \bullet$] \textsf{Encrypt}$(\pi,sk,\bm{m})$:   To encrypt a message $\bm{m} \in \{0,1\}^{\ell-n}$, sample a vector $\bm{y}$, uniformly at random from $\mathbb{Z}_q^n$ and a vector $\bm{e}=(\bm{0},e_{n+1},\hdots,e_{\ell}) \in \mathbb{Z}_q^\ell$, where $\bm{0}$ denotes the zero vector of order $n$ and each $e_j$ is chosen independently from the distribution $\mathcal{X}$ for $n+1 \leq j \leq \ell$. Let $\bm{p}$ be the vector given by $ \bm{p}=(\bm{0},{\bm{m}})=(\bm{0},m_{n+1},\hdots,m_{\ell})\in \mathbb{Z}_q^\ell$. If $\bm{S}_{\text{enc}}:=\begin{bmatrix}\bm{I}_n &  -\bm{S}^T \end{bmatrix} \in \mathbb{Z}_q^{n \times \ell}$, then the ciphertext can be computed as:
{ \begin{align}
\textstyle \bm{c}=\textstyle  \left( \bm{p}\cdot \left\lfloor \frac{q}{2} \right\rfloor+ \bm{y} \cdot \bm{S}_{\text{enc}}+  \bm{e} \right)\bm{R} ~{mod}~q \in \mathbb{Z}_q^\ell
\end{align}}%
Note that, $\bm{y} \cdot \bm{S}_{\text{enc}}$ essentially represents the evaluation of a polynomial $f \in \mathcal{I}_{\leq r}$ at the points $(\bm{z}_1,\hdots,\bm{z}_\ell)$. This is because $f(\bm{z}_1),\hdots,f(\bm{z}_n)$ can take any value $\bm{y} \in \mathbb{Z}_q^n$ and $f(\bm{z}_j)$ for $n+1 \leq j \leq \ell$ can be expressed as a linear combination of $f(\bm{z}_1),\hdots,f(\bm{z}_n)$, i.e.,
\vspace{-1em}
\begin{align}\label{evaluations}
    f(\bm{z}_j)=-\sum_{i=1}^n s_j^i \cdot f(\bm{z}_i)~(mod~q)
\end{align}

\item[$  \bullet$] \textsf{Decrypt}$(\pi,sk,\bm{c})$: Let $\bm{S}_{\text{dec}}=\bm{R}^{-1}\begin{bmatrix} \bm{S} &~ \bm{I}_{\ell-n} \end{bmatrix}^T \in \mathbb{Z}_q^{\ell \times (\ell-n) }$. Given the ciphertext $\bm{c}$ and the secret key $sk$, ${\bm{m}}$ can be recovered as
{ \begin{align}\label{decrypt}
 {\bm{m}}=\left\lfloor \frac{1}{\left\lfloor q/2 \right\rfloor}\left(\bm{c}\cdot\bm{S}_{\text{dec}} \,{mod}\,q\right)\right\rceil\,{mod}~2 
\end{align} }%
   \end{itemize}

\subsubsection{Correctness of Decryption.} \label{COD}
For correct decryption, the noise in the decryption process must be small. The decryption process involves computing the inner product of the ciphertext with each column of the matrix $\bm{S}_{\text{dec}}$ and reducing it modulo $q$. For each of these products, the decryption function outputs 0 for the corresponding entry if the magnitude of the inner product is $<q/4$ and 1 otherwise. %The columns of $\bm{S}_d$ are the vectors $\tilde{\bm{s}}_j^T$ from Equation (\ref{sk}) for $n+1 \leq j \leq \ell$. 
In the following lemma, we analyze the magnitude of the noise in decryption. 
\begin{lemma} \label{CoE}
Let $q,n,\ell,\abs{\mathcal{X}} \leq B$ be as described in the scheme. Let $\bm{c}=\left( \bm{p} \cdot \left\lfloor \frac{q}{2} \right\rfloor+ \bm{y} \cdot \bm{S}_{\textup{enc}}+  \bm{e} \right)\bm{R} ~{mod}~q$ be the encryption of a message $\bm{m} \in\{0,1\}^{\ell-n}$ under the key $sk=(\bm{S},\bm{R}_1,\bm{R}_2)$. Then, for some $\bm{e}=(\bm{0},\tilde{\bm{e}}) \in (\bm{0},\mathcal{X}^{\ell-n}) $ with $\norm{{\bm{e}}}_{\infty} \leq B$, it holds that %\begin{center}
$  \bm{c}\cdot  \bm{S}_{\textup{dec}} ={\bm{m}} \cdot \left\lfloor \frac{q}{2} \right\rfloor + \tilde{\bm{e}}~({mod}~q)$.
%\end{center} 
Further, if $B < \left\lfloor {q}/{2} \right\rfloor/2 $, then ${\bm{m}} \leftarrow \textup{\textsf{Decrypt}}(sk,\bm{c})$.
\end{lemma}
\begin{proof}
%It can be easily verified that $\bm{R}^{-1} f(\bm{Z}) \in \mathcal{I}_{\leq r}$. 
%Observe that, $\tilde{\bm{S}}^T(:,j)=\tilde{\bm{s}}_j$ for $1 \leq j \leq \ell-n$. 
%Observe that, $\bm{S}_{e} \cdot \bm{S}_{d}=0^{(\ell-n)\times n}~mod~q$. Therefore, $\bm{y}\cdot \bm{S}_{e}\cdot\bm{S}_{d}=\bm{0}~({mod}~q)$.
If $\tilde{\bm{e}}=(e_{n+1},\hdots,e_{\ell}) \in \mathbb{Z}_q^{\ell-n}$ where $e_j$ denotes the $j^{th}$ non-zero entry of $\bm{e}$ for $n+1 \leq j \leq \ell$, then
\begin{align}
\bm{c}\cdot  \bm{S}_{\textup{dec}} &= \left( \bm{p} \cdot \left\lfloor \frac{q}{2} \right\rfloor +\bm{y}\cdot \bm{S}_{\text{enc}}+\bm{e} \right)\bm{R} \cdot \bm{S}_{\text{dec}} ~({mod}~q)\notag\\
%&={\bm{m}} \cdot \left\lfloor \frac{q}{2} \right\rfloor+ \bm{e}\cdot \bm{S}_{d} ~({mod}~q) \notag\\
&= {\bm{m}} \cdot \left\lfloor \frac{q}{2} \right\rfloor+ \tilde{\bm{e}}~({mod}~q)
\end{align} 

If $B <  \left\lfloor {q}/{2} \right\rfloor/2$, then for $m_j=0$ for any $n+1 \leq j \leq \ell$ , $\bm{c}\cdot  \bm{S}_{\textup{dec}}(j) ={e}_{j}~({mod}~q)$ where $\abs{e}_j \leq B < q/4$. Hence, the decryption function outputs 0. Similarly, for $m_j=1$, $ \bm{c}\cdot  \bm{S}_{\textup{dec}}(j)=  \left\lfloor \frac{q}{2} \right\rfloor+ e_j~({mod}~q) $ whose magnitude is $> q/4$ and hence, the decryption function outputs 1. Therefore, if $\norm{{\bm{e}}}_{\infty} \leq B$ and $B <  \left\lfloor {q}/{2} \right\rfloor/2$, then $\bm{m} \leftarrow \textup{\textsf{Decrypt}}(sk,\bm{c})$.
\qed
\end{proof}

\subsection{Security}  \label{security}

The security of this scheme follows from Lemma 6.2 of \cite{peikert2011lossy}. The lemma is restated as follows
\begin{lemma} \label{Peikart}
Let $h,\ell =poly(n)$. Choose $\bm{A} \leftarrow \mathbb{Z}_q^{h \times n}$, $\bm{\hat{S}}\leftarrow \mathbb{Z}_q^{ (\ell-n) \times n}$ uniformly at random and $\bm{E} \leftarrow \mathcal{X}^{h \times (\ell-n)}$. If $\bm{B}=\bm{A}\bm{\hat{S}}^T+\bm{E}$, then $(\bm{A},\bm{B})$ is computationally indistinguishable from uniform over $\mathbb{Z}_q^{h \times \ell}$ under the assumption that LWE$_{n,q,\mathcal{X}}$ is hard.
\end{lemma}

To prove the  security of the proposed scheme we consider  $h=1$. 
%The security of the proposed scheme follows directly from the security of the multiple secret scheme described in \cite{peikert2008framework}. In this scheme, the encryption of a message vector $\bm{m} =(m_1,m_2,\ldots,m_{\ell-n})\in \mathbb{Z}_q^{\ell-n}$ under a secret key $\bm{S} \in \mathbb{Z}_q^{(\ell-n)\times n}$ is given by 
%\begin{eqnarray*}
%\bm{c} = (\bm{v},m_1 +\langle \bm{s}_1,\bm{v}\rangle +e_1,m_2 +\langle\bm{s}_2,\bm{v}\rangle+e_2 ,\ldots,m_{\ell-n} +\langle \bm{s}_{\ell-n},\bm{v}\rangle+e_{\ell-n})
%\end{eqnarray*}
%where $\bm{v} \in \mathbb{Z}_q^n$ is randomly chosen from a uniform distribution and the $\bm{s}_i$s are the rows of $\bm{S}$. The secret key $\bm{S}$ is randomly chosen from a uniform distribution in $ \mathbb{Z}_q^{(\ell-n)\times n}$. The security of this scheme (the scheme given in \cite{peikert2008framework}) is based on the hardness of the LWE problem.

\begin{lemma}
Under the LWE assumption, given two distinct message vectors $\bm{m}_1,\bm{m}_2 \in \{0,1\}^{\ell-n}$, if $\ell-n $ is $\mathcal{O}(1)$ there exists no efficient algorithm that can distinguish between the distributions of the encryptions of $\bm{m}_1$ and $\bm{m}_2$. Moreover, there exists no efficient algorithm that can distinguish the uniform distribution on the set of encryptions of any given message from the uniform distribution on $\mathbb{Z}_q^\ell$
\end{lemma}
\begin{proof}
Note that the matrix $\bm{R}$ in the encryption process of the proposed scheme is the product of the matrices $\bm{R}^{\prime} = \left[ \begin{matrix}\bm{R}_1 & \bm{0}\\\bm{0} &\bm{I} \end{matrix} \right]$ and $\bm{R}^{\prime \prime} = \left[ \begin{matrix}\bm{I} & \bm{0}\\\bm{R}_2 &\bm{I} \end{matrix} \right]$.
Therefore, the encryption of a message $\bm{m}$ is given as
\begin{eqnarray*}
\bm{c} = \left( \bm{p} \cdot\left\lfloor \frac{q}{2} \right\rfloor  + \bm{y}\cdot\bm{S}_{\text{enc}} + \bm{e}\right)\bm{R}^{\prime}\bm{R}^{\prime \prime}
\end{eqnarray*}
where $\bm{p} = (\bm{0},\bm{m})$ and $\bm{y}$ is randomly chosen from a uniform distribution in $\mathbb{Z}_q^n$. 
Now, if we substitute $\bm{y}\bm{R}_1 = \bm{v}$, we get $\bm{c}^{\prime}=( \bm{p}\lfloor \frac{q}{2} \rfloor  + \bm{y}\cdot\bm{S}_{\text{enc}} + \bm{e})\bm{R}^{\prime} = (\bm{p}\lfloor \frac{q}{2}\rfloor + (\bm{v},-\bm{v}\cdot \bm{R}_1^{-1}\bm{S}^T) +  \bm{e})$. Here, the distribution on $\bm{v}$ is uniform in $\mathbb{Z}_q^n$ and the distribution on $\bm{R}_1^{-1}\bm{S}^T$ is uniform in the set of full column rank matrices in $\mathbb{Z}_q^{n \times (\ell-n)}$. Observe that the distribution on $((\bm{v},-\bm{v}\cdot \bm{R}_1^{-1}\bm{S}^T) +  \bm{e})$ is similar to the one considered in Lemma \ref{Peikart} (for the case $h=1$) except that there the distribution on $\bm{\hat{S}}$ is random in $\mathbb{Z}_q^{(\ell-n)\times n}$ (while here the distribution of $\bm{S}\bm{R}_1^{-T}$ is restricted to full row rank matrices in $\mathbb{Z}_q^{(\ell-n)\times n}$).

%The distribution of $\bm{c}_i^{\prime}=( \bm{p}_i\lfloor \frac{q}{2} \rfloor  + \bm{y}_i\cdot\bm{S}_{\text{enc}} + \bm{e}_i)\bm{R}^{\prime}$ is similar to the distribution of the encryption of $\bm{m}_i$ in the multiple secret scheme described in \cite{peikert2008framework}. The only difference being that in the proposed scheme the distribution of $\bm{S}\bm{R}_1^{-1}$ is uniform in the space of full row rank matrices in $\mathbb{Z}^{(\ell-n) \times n}$ while the distribution of the secret key in the scheme in \cite{peikert2008framework} is uniform in the set of all matrices in $\mathbb{Z}^{(\ell-n) \times n}$.
When $\ell -n$ is $\mathcal{O}(1)$,  the cardinality of the set of full rank matrices in $\mathbb{Z}_q^{ (\ell-n)\times n}$ is a significant fraction $( \frac{1}{\mathcal{O}(1)})$ of the cardinality of the set of all matrices in $\mathbb{Z}_q^{(\ell-n)\times n}$. Therefore,  if there exists an algorithm that can efficiently distinguish between the distribution of $((\bm{v},-\bm{v}\cdot \bm{R}_1^{-1}\bm{S}^T) +  \bm{e}) = ( \bm{y}\cdot\bm{S}_{\text{enc}} + \bm{e})\bm{R}^{\prime}$ from the uniform distribution on $\mathbb{Z}_q^\ell$ for a non-negligible fraction of choices of $\bm{S}$ and $\bm{R}_1$, it can also distinguish  the distribution on $(\bm{A},\bm{B})$ in Lemma \ref{Peikart} from the uniform one for a non-negligible fraction of $\bm{\hat{S}}$ (Here, non-negligible means 
$( \frac{1}{\mathcal{O}(n^c)})$ for some constant $c$).

Now, suppose there exist a non-zero message $\bm{m}$ and an algorithm $W$ that can distinguish the corresponding distribution of $\bm{c}^{\prime}=( \bm{p}\lfloor \frac{q}{2} \rfloor  + \bm{y}\cdot\bm{S}_{\text{enc}} + \bm{e})\bm{R}^{\prime}$ from the distribution of $(\bm{y}\cdot\bm{S}_{\text{enc}} + \bm{e})\bm{R}^{\prime}$ (for a non-negligible fraction of choices of $\bm{S}_{\text{enc}}$ and $\bm{R}^{\prime}$) then such an algorithm can be used to create an algorithm $W^{\prime}$ that distinguishes between the distribution on $(\bm{y}\cdot\bm{S}_{\text{enc}} + \bm{e}_i)\bm{R}^{\prime}$ and the uniform one. Let $p_W(\bm{m})$ be the probability that $W$ returns 1 when the input is sampled from the distribution on  $\bm{c}^{\prime}= ((0,\bm{m})+ \bm{y}\cdot\bm{S}_{\text{enc}} + \bm{e})\bm{R}^{\prime}$ and let $p_W(\bm{0})$ be the probability that $W$ returns 1 when the input is sampled from the distribution on $(\bm{y}\cdot\bm{S}_{\text{enc}} + \bm{e})\bm{R}^{\prime}$. Suppose $|p_W(\bm{0})-p_W(\bm{m})| > \epsilon$ for some significant value $\epsilon$. Then, either $|p_W(\bm{0})-p_W(\bm{U})|$ or $|p_W(\bm{m})-p_W(\bm{U})|$ must be greater than $\frac{\epsilon}{2}$, where $p_W(\bm{U})$ is the probability that $W$ returns 1 when the input is sampled from the uniform distribution. If $|p_W(\bm{0})-p_W(\bm{U})|>\frac{\epsilon}{2}$ then 
$W^{\prime}$ is identical to $W$. Otherwise, $W^{\prime}$ calls the algorithm $W$ after altering its input by adding $(\bm{0},\bm{m})$ to it. (These arguments are similar to the arguments in the proof Lemma 5.4 in \cite{regev2009lattices}.) Let $\bm{c}_{\bm{m}}$ denote the encryption of a message $\bm{m}$ under the secret key $\bm{S},\bm{R}_1,\bm{R}_2$. The above arguments prove that the probability of distinguishing the distribution of $\bm{c}_{\bm{m}}(\bm{R}^{\prime\prime})^{-1}$ from that of $\bm{c}_{\bm{0}}(\bm{R}^{\prime\prime})^{-1}$ for any efficient algorithm is negligible under the LWE assumption (The probability is taken over the choices of $\bm{S}$ and $\bm{R_1}$ and the randomness involved in the encryption process).

Since $\bm{R}_2$ is chosen randomly from a uniform distribution, the above arguments imply that, under the LWE assumption, there exists no efficient algorithm that can distinguish between encryptions of a non zero message $\bm{m}$ from the encryptions of the $\bm{0}$ vector. Now, for an algorithm $\hat{W}$, let $p_{\hat{W}}(\bm{m})$ denote the probability of $\hat{W}$ returning $1$ when the input is sampled from the distribution on the encryptions of $\bm{m}$. Clearly, for two distinct message vectors $\bm{m}_1$ and $\bm{m}_2$
\begin{equation}
    |p_{\hat{W}}(\bm{m}_1)-p_{\hat{W}}(\bm{m}_2)| \leq |p_{\hat{W}}(\bm{m}_1)-p_{\hat{W}}(\bm{0})| + |p_{\hat{W}}(\bm{m}_2)-p_{\hat{W}}(\bm{0})|
\end{equation}
 Since, under the LWE assumption, both the terms on the right hand side of the above equation are negligible for all efficient algorithms, the value of the term on the left hand side is also negligible. Further, since there exists no efficient algorithm $W$ such that $|p_W(\bm{0})-p_W(\bm{U})|$ is non-negligible, there exists no efficient algorithm $\hat{W}$ such that $|p_{\hat{W}}(\bm{0})-p_{\hat{W}}(\bm{U})|$ is non-negligible. For any message $\bm{m}$ and algorithm $\hat{W}$
 
 \begin{equation}
    |p_{\hat{W}}(\bm{m})-p_{\hat{W}}(U)| \leq |p_{\hat{W}}(\bm{m})-p_{\hat{W}}(\bm{0})| + |p_{\hat{W}}(\bm{0})-p_{\hat{W}}(U)|.
\end{equation}
 Therefore, under the LWE assumption, there exists no efficient algorithm $\hat{W}$ such that $|p_{\hat{W}}(\bm{m})-p_{\hat{W}}(U)|$ is non-negligible. In other words, there exists no efficient algorithm that can distinguish the distribution on the encryptions of a message $\bm{m}$ from the uniform distribution on $\mathbb{Z}_q^\ell$.
\qed 
  
\end{proof} 

\subsection{Homomorphic Properties}\label{HP}
 
  %If \(\phi:\{0,1\}^t \rightarrow \{0,1\}$ denotes a function to be performed on the plaintexts $m_1,\hdots,m_t$, then homomorphically evaluating $\phi$ involves calculating a new ciphertext \(\bm{c}_{\scriptscriptstyle eval} \) such that \(Dec(\bm{c}_{\scriptscriptstyle eval},sk)=\phi(m_1,\hdots,m_t)\). Any function $\phi$ involving binary variables can be evaluated using a set of addition and multiplication gates. 
  The proposed scheme can be used to homomorphically evaluate a function $\phi:\{0,1\}^{\tau(\ell-n)} \rightarrow \{0,1\}^{\ell-n}$ on ciphertexts $\bm{c}_1,\hdots,\bm{c}_\tau$ such that $\phi(\bm{c}_1,\hdots,\bm{c}_\tau)$ yields a ciphertext $\bm{c}_{\phi}$. In the proposed scheme, $\phi$ represents an arithmetic circuit over $GF(2)$ with addition and multiplication gates. We now show how to perform homomorphic addition and multiplication of two ciphertexts in the proposed scheme.

\subsubsection{Addition.} Addition is performed by simply adding the ciphertexts. For some $\bm{y}_1,\bm{y}_2 \in \mathbb{Z}_q^n$, if \(\bm{c}_1=\left(\bm{p}_1 \left\lfloor \frac{q}{2} \right\rfloor+\bm{y}_1 \cdot \bm{S}_{\text{enc}}+\bm{e}_1\right)\bm{R}~{mod}~q\) and \(\bm{c}_2= \left(\bm{p}_2\left\lfloor \frac{q}{2} \right\rfloor+\bm{y}_2 \cdot \bm{S}_{\text{enc}}+\bm{e}_2\right)\bm{R}~{mod}~q\) denote the respective encryptions of $\bm{m}_1$ and $\bm{m}_2$, then compute 
\begin{align}
 \bm{c}_{ add} &=\bm{c}_1+  \bm{c}_2~~({mod}~q) \notag \\
  & = \left((\bm{p}_1+\bm{p}_2)\left\lfloor \frac{q}{2} \right\rfloor+(\bm{y}_1+\bm{y}_2)\cdot \bm{S}_{\text{enc}}+\bm{e}_1+\bm{e}_2\right)\bm{R}~({mod}~q)
\end{align}
where $\bm{e}_i=(\bm{0},\tilde{\bm{e}}_{i})$ for some $ \tilde{\bm{e}}_{i} \leftarrow \mathcal{X}^{\ell-n} $ for $i \in \{1,2\}$. %Let $\tilde{\bm{e}}_{i}=(e_{i,n+1},\hdots,e_{i,\ell})$ where $e_{i,j}$ denote the non-zero entries of $\bm{e}_i$ for $i \in \{1,2\}$ and $n+1 \leq j \leq \ell$. 
If $\tilde{\bm{e}}_{add}:=\tilde{\bm{e}}_1+\tilde{\bm{e}}_2$, then %. Since, $ \bm{S}_{enc}\cdot \bm{S}_{dec}={0}~(mod~q)$,  we have 
{\begin{align}
\bm{c}_{ add}\cdot\bm{S}_{\text{dec}} &=({\bm{m}}_1+{\bm{m}}_2) \left\lfloor \frac{q}{2} \right\rfloor+ \tilde{\bm{e}}_{add}~({mod}~q) \notag \\
 &=({\bm{m}}_1 \oplus {\bm{m}}_2) \left\lfloor \frac{q}{2} \right\rfloor -  \frac{1}{2}  [{\bm{m}}_1+{\bm{m}}_2-({\bm{m}}_1 \oplus {\bm{m}}_2)] +\tilde{\bm{e}}_{add}~({mod}~q) 
\end{align}}
 If \(\bm{e}_{add}:=-  \frac{1}{2}  [{\bm{m}}_1+{\bm{m}}_2-({\bm{m}}_1 \oplus {\bm{m}}_2)] +\tilde{\bm{e}}_{add}\) and $\norm{{\bm{e}}_1}_{\infty}, \norm{{\bm{e}}_2}_{\infty} \leq B $, then the magnitude of the noise after addition can be computed as
 \begin{align}
 \norm{\bm{e}_{add}}_{\infty}=\norm{-  \frac{1}{2}  [{\bm{m}}_1+{\bm{m}}_2-({\bm{m}}_1 \oplus {\bm{m}}_2)]+( \tilde{\bm{e}}_1+\tilde{\bm{e}}_2)}_{\infty} \leq 1+2B
 \end{align}

\subsubsection{Multiplication.} \label{mult} 

Given two ciphertexts $\bm{c}_1$ and $\bm{c}_2$ that encrypts the messages $\bm{m}_1$ and $\bm{m}_2$, homomorphic multiplication is performed by using a bilinear map on $\bm{c}_1$ and $\bm{c}_2$. %$\mathcal{B}_{\mathbfcal{M}} : \mathbb{Q}^{\ell} \times \mathbb{Q}^{\ell} \rightarrow \mathbb{Q}^\ell$ such that $\mathcal{B}_{\mathbfcal{M}}(\bm{c}_1,\bm{c}_2)\approx\textsf{Encrypt}(\bm{m}_1 \odot \bm{m}_2,sk)$. 
This map is represented by a 3-way tensor $\mathbfcal{M}$ which is provided as the public evaluation key for multiplication. We now proceed to construct $\mathbfcal{M}$. All operations are performed over  $\mathbb{Q}$  unless stated otherwise. 
For some $x \in \mathbb{Q}$, $y=x~({mod}~q)$ denotes the unique value in the interval $(-q/2,q/2]$.  

Homomorphic multiplication in this scheme uses the fact that given two polynomials $f_1,f_2 \in \mathcal{I}_{\leq r}$ and an evaluation point $\bm{z} \in \mathbb{Z}_q^v$, 
%{ \begin{align}
$f_1(\bm{z}) \cdot f_2(\bm{z})=\left(f_1 \cdot f_2\right) (\bm{z})$.
%\end{align}}%
Although $f_1f_2 \in \mathcal{I}$ it need not be an element of the subspace $\mathcal{I}_{\leq r}$. Instead, it is an element of the space $\mathcal{I}_{\leq 2r}$. Let $n_1$ denote the dimension of the subspace $\mathcal{I}_{\leq 2r}$ and $t  = n_1 +\ell -n$. We choose $(n_1-n)$ additional points $\bm{z}_{\ell+1},\hdots,\bm{z}_t$ in $\mathbb{Z}_q^v$ such that every vector in $\mathbb{Z}_q^{n_1}$ can be obtained by evaluating a polynomial in $\mathcal{I}_{\leq 2r}$ at $(\bm{z}_1,\hdots,\bm{z}_n,\bm{z}_{\ell+1},\hdots,\bm{z}_t)$. Evaluating polynomials in  $\mathcal{I}_{\leq 2r}$ on the points $\bm{z}_1,\bm{z}_2,\ldots,\bm{z}_t$ yields an $n_1$ dimensional subspace of $\mathbb{Z}_q^t$

 %. For $n+1 \leq j \leq \ell$, we wish to multiply the components of $\tilde{\bm{c}}_1$ and $\tilde{\bm{c}}_2$ using a similar technique proposed in \cite{brakerski2012fully}. %The ciphertext is scaled by a factor of $q$ and as a result, homomorphic multiplication increases the noise only by a polynomial factor in the security parameter $\lambda$. %Similar to \cite{brakerski2012fully}, the ciphertext $\tilde{\bm{c}}$ is scaled by $q$.  

Let $f_{mult}$ be a polynomial in $\mathcal{I}_{\leq r} $ and $(f_{mult}(\bm{z}_1),\hdots,f_{mult}(\bm{z}_\ell))=\bm{y}_{mult} \cdot \bm{S}_{\text{enc}}$ for some $\bm{y}_{mult} \in \mathbb{Z}_q^n$. Then, given encryptions of $\bm{m}_1$ and $\bm{m}_2$ viz. $\bm{c}_1$ and $\bm{c}_2$, our aim is to get a ciphertext of the form 
%\vspace{-1ex}
{ \begin{align}
\bm{c}_{mult}=\left((\bm{p}_1 \odot \bm{p}_2)\cdot \left\lfloor \frac{q}{2} \right\rfloor+ \bm{y}_{mult}\cdot \bm{S}_{\text{enc}}+\bm{e}_{mult} \right)\bm{R}~mod~q       %=\bm{R}_1 \cdot \begin{bmatrix}
% f^{mult}(\bm{z}_1) \\
% \vdots \\
% f^{mult}(\bm{z}_n) \\
% (m_1m_2)\cdot \left\lfloor \frac{q}{2} \right\rfloor+ f^{mult}(\bm{z}_{n+1})+{e}_{mult,n+1} \\
% \vdots \\
% (m_1m_2)\cdot \left\lfloor \frac{q}{2} \right\rfloor+ f^{mult}(\bm{z}_{\ell})+ {e}_{mult,\ell} 
% \end{bmatrix} 
\in \mathbb{Z}_q^\ell
\end{align}}

For $i := \{1,2\}$, ciphertexts $\bm{c}_i$ that encrypt messages $\bm{m}_i$ are given as
%let $f_i(\bm{Z})=(f_i(\bm{z}_1),\hdots,f_i(\bm{z}_\ell))$ be the evaluation of polynomials $f_i \in \mathcal{I}_{\leq r}$ at $(\bm{z}_1,\hdots,\bm{z}_\ell)$. Then,  ciphertexts $\bm{c}=(\bm{p}\cdot \left\lfloor {q}/{2} \right\rfloor+\bm{y}\cdot \bm{S}_{e}+\bm{e})\bm{R}_1~mod~q$ that encrypts a message $\bm{m} \in \{0,1\}^{\ell-n}$ can be written as
\begin{align}
    \bm{c}_i = \!\begin{bmatrix} \begin{matrix}
    f_i(\bm{z}_1) \\ \vdots \\ f_i(\bm{z}_n)\end{matrix} \\ 
    m_{i,n+1}  \left\lfloor \frac{q}{2} \right\rfloor +f_i(\bm{z}_{n+1}) +e_{i,n+1} \\
    \vdots \\
    m_{i,\ell}  \left\lfloor \frac{q}{2} \right\rfloor +f_i(\bm{z}_{\ell}) +e_{i,\ell}
    \end{bmatrix}^T\bm{R}~(mod~q)
\end{align}
where $f_i$s are polynomials that are randomly sampled from $\mathcal{I}_{\leq r}$.

The process of homomorphically multiplying $\bm{c}_1$ and $\bm{c}_2$ is done through the following steps. %The tensor $\mathbfcal{M}$ acting on $\bm{c}_1$ and $\bm{c}_2$ performs the following operations.
\begin{enumerate}
\setlength\itemsep{1ex}
\item For some $K_{1,j},K_{2,j} \in \mathbb{Z}$, transform the ciphertexts $\bm{c}_1$ and $\bm{c}_2$ to  vectors of the form,

%\vspace{-1ex}
{\small  \begin{flalign} \label{step1}
& \begin{aligned} 
\tilde{\bm{c}}_1 \!=\!\!\begin{bmatrix}
f_1(\bm{z}_1) \\
\vdots \\
f_1(\bm{z}_n)\\
m_{1,n+1}  \left\lfloor \frac{q}{2} \right\rfloor  +e_{1,n+1}^{\prime}+qK_{1,n+1} \\
\vdots \\
m_{1,\ell} \left\lfloor \frac{q}{2} \right\rfloor +e_{1,\ell}^{\prime} +qK_{1,\ell}
\end{bmatrix}^T
\end{aligned},
& \!\!\! 
\!\! \begin{aligned}
\!\! \tilde{\bm{c}}_2\!=\!\!\begin{bmatrix}
f_2(\bm{z}_1) \\
\vdots \\
f_2(\bm{z}_n)\\
m_{2,n+1} \left\lfloor \frac{q}{2} \right\rfloor+e_{2,n+1}^{\prime}+qK_{2,n+1} \\
\vdots \\
m_{2,\ell} \left\lfloor \frac{q}{2} \right\rfloor+e_{2,\ell}^{\prime}+qK_{2,\ell} 
\end{bmatrix}^T\!\!\in \mathbb{Q}^\ell
\end{aligned}
\end{flalign}}
(Note that the noise terms $e_{i,j}$ have transformed to $e_{i,j}^{\prime}$. This is because we deliberately introduce some noise in this step. The reason for the same is explained later in the paper.)
\item For some $K_{1,j},K_{2,j} \in \mathbb{Z}$ for $\ell+1 \leq j \leq t$, compute values that are equivalent mod $q$ to evaluations of $f_1$ and $f_2$ at the additional points, $\bm{z}_{\ell+1},\hdots,\bm{z}_t$ and append these entries to $\tilde{\bm{c}}_i$ to generate vectors %. To include $(t-\ell)$ additional evaluations of $f_1$ and $f_2$ at $(\bm{z}_{\ell+1},\hdots,\bm{z}_t)$, extend $ \tilde{\bm{c}}_1 $ and $ \tilde{\bm{c}}_2 $ to the vectors 
$\bm{c}_1^{\prime}$ and $\bm{c}_2^{\prime}$ where

\vspace{-1ex}
{\small  \begin{flalign} \label{step2}
& \begin{aligned} 
{\bm{c}}_1^{\prime} \!=\!\!\begin{bmatrix}
f_1(\bm{z}_1) \\
\vdots \\
f_1(\bm{z}_n)\\
m_{1,n+1}  \left\lfloor \frac{q}{2} \right\rfloor  +e_{1,n+1}^{\prime}+qK_{1,n+1} \\
\vdots \\
m_{1,\ell} \left\lfloor \frac{q}{2} \right\rfloor +e_{1,\ell}^{\prime} +qK_{1,\ell}\\
f_1(\bm{z}_{\ell+1})+qK_{1,\ell+1} \\
\vdots \\
f_1(\bm{z}_t)+qK_{1,t}
\end{bmatrix}^T
\end{aligned},
& \!\!\! 
\!\! \begin{aligned}
\!\! {\bm{c}}_2^{\prime}\!=\!\!\begin{bmatrix}
f_2(\bm{z}_1) \\
\vdots \\
f_2(\bm{z}_n)\\
m_{2,n+1} \left\lfloor \frac{q}{2} \right\rfloor+e_{2,n+1}^{\prime}+qK_{2,n+1} \\
\vdots \\
m_{2,\ell} \left\lfloor \frac{q}{2} \right\rfloor+e_{2,\ell}^{\prime}+qK_{2,\ell} \\
f_2(\bm{z}_{\ell+1})+qK_{2,\ell+1} \\
\vdots \\
f_2(\bm{z}_t)+qK_{2,t}
\end{bmatrix}^T \in \mathbb{Q}^t
\end{aligned}
\end{flalign}}

\item Take component-wise product of $\bm{c}_1^{\prime}$ and $\bm{c}_2^{\prime}$ and multiply the entries containing the message with $2/q$ to get  %. Construct a tensor $\mathbfcal{P}$ to take component-wise product of $(\tilde{\bm{c}}_1,f_1(\bm{z}_{\ell+1}),\hdots,f_1(\bm{z}_t))$ and $(\tilde{\bm{c}}_2,f_2(\bm{z}_{\ell+1}),\hdots,f_2(\bm{z}_t))$.  

{\small \begin{align} \label{cmult_dash}
\bm{c}_{mult}^{\prime}\!=\!\!\begin{bmatrix}
f_1f_2(\bm{z}_1) \\
\vdots \\
f_1f_2(\bm{z}_n)\\
\frac{2}{q}\! \left(m_{1,n+1} \left\lfloor \frac{q}{2} \right\rfloor \!+{e}_{1,n+1}^{\prime}+\!qK_{1,n+1}\right)\!\left(m_{2,n+1} \left\lfloor \frac{q}{2} \right\rfloor \!+ {e}_{2,n+1}^{\prime}+\!qK_{2,n+1}\!\right) \!\\
\vdots \\
\frac{2}{q}  \left(m_{1,\ell} \left\lfloor \frac{q}{2} \right\rfloor +{e}_{1,\ell}^{\prime}+qK_{1,\ell}\right)\!\left(m_{2,\ell} \left\lfloor \frac{q}{2} \right\rfloor + {e}_{2,\ell}^{\prime}+qK_{2,\ell}\right) \\
f_1f_2(\bm{z}_{\ell+1})+qK_{\ell+1} \\
\vdots \\
f_1f_2(\bm{z}_{t})+qK_{t}
\!\!\end{bmatrix}^T 
\end{align}}
where $K_j \in \mathbb{Z}$ for $\ell+1 \leq j \leq t$.

\item Add integers equivalent to $f_1f_2(\bm{z}_{j})~mod~q$ to the $j^{th}$ entries of $\bm{c}_{mult}^{\prime}$ for $n+1 \leq j \leq \ell$. Let the resultant vector be $ \bm{c}_{mult}^{\prime\prime} \in \mathbb{Q}^t$ such that for $n+1 \leq j \leq \ell$ 

{\small  \begin{align} \label{c_mult2}
 \bm{c}_{mult}^{\prime\prime}(j)= \frac{2}{q}  \left(m_{1,j} \left\lfloor \frac{q}{2} \right\rfloor +{e}_{1,j}^{\prime}+qK_{1,j}\right)\!\left(m_{2,j} \left\lfloor \frac{q}{2} \right\rfloor + {e}_{2,j}^{\prime}+qK_{2,j}\right)+f_1f_2(\bm{z}_{j})+qK_j
\end{align}}
where $K_j \in \mathbb{Z}$ for $n+1 \leq j \leq \ell$. The remaining entries of $ \bm{c}_{mult}^{\prime\prime} $ are the same as that of $ \bm{c}_{mult}^{\prime} $.

\item Transform the vector $ \bm{c}_{mult}^{\prime\prime} \in \mathbb{Q}^t$ to a valid ciphertext $\bm{c}_{mult}$ of size $\ell$ over $\mathbb{Z}_q$.

\end{enumerate}

We now explain in detail how each of the above steps is performed. 

\vspace{1ex}

\textbf{Step 1:} The first step is to transform the ciphertexts $\bm{c}_1$ and $\bm{c}_2$ to the vectors $\tilde{\bm{c}}_1$ and $\tilde{\bm{c}}_2$ given in Equation (\ref{step1}). For $i:=\{1,2\}$, let $\bm{D}_i$s be matrices given by
\begin{align}
\bm{D}_i&=\bm{R}^{-1} \cdot \left[
\begin{array}{c|c}
{\bm{I}}_n &~ \bm{S}^T  \\
\hline
\bm{0} &~ \bm{I}_{\ell-n}
\end{array}
\right] + \left[\begin{array}{c|c}
\bm{0} &~ \bm{\epsilon}_i \\
\hline
\bm{0} &~ \bm{0}
\end{array}\right]\\
&=\left[
\begin{array}{c|c}
{\bm{R}_1}^{-1} &~ {\bm{R}_1}^{-1}\bm{S}^T  + \bm{\epsilon}_i \\
\hline
-\bm{R}_2{\bm{R}_1}^{-1} &~ -\bm{R}_2{\bm{R}_1}^{-1}\bm{S}^T +\bm{I}_{\ell-n}
\end{array}
\right] \in \mathbb{Q}^{\ell \times \ell}
\end{align}
where $\bm{\epsilon}_i$s are matrices in $\mathbb{Q}^{n \times (\ell-n)}$ such that the one norm of each of  their columns is less that ${B}/{q}$. Observe that, for $i:=\{1,2\}$ and $n+1 \leq j \leq \ell$, $\langle \bm{c}_i\bm{R}^{-1},\tilde{\bm{s}}_j \rangle =m_{i,j} \left\lfloor \frac{q}{2} \right\rfloor +{e}_{i,j}+qK_{i,j}$, where $K_{1,j}, K_{2,j} \in \mathbb{Z} $ and the $\tilde{\bm{s}}_j $s are as given in Equation (\ref{sk}). Therefore, for $i:=\{1,2\}$,  $\tilde{\bm{c}}_i = \bm{c}_i\bm{D}_i$  and the error terms $e_{i,j}^{\prime}$ are given by
$e_{i,j}^{\prime} = e_{i,j} + \langle \bm{c}_i, (\bm{\epsilon}_i(:,j),\bm{0}) \rangle$ where $\bm{0} \in \mathbb{Q}^{\ell-n}$ is the all zero vector. Both terms in the right hand side of this equation are bounded by $B$.

\vspace{1em}

\textbf{Step 2:}
%The vectors $ \tilde{\bm{c}}_1 $ and $ \tilde{\bm{c}}_2 $ can now be extended to $\bm{c}_1^{\prime}$ and $\bm{c}_2^{\prime}$ as follows. 
Evaluation of polynomials in $\mathcal{I}_{\leq r}$ at $\bm{z}_1,\hdots,\bm{z}_t$ constitutes an $n$-dimensional subspace of $\mathbb{Z}_q^t$. Therefore, for $i \in \{1,2\}$ and $\ell+1 \leq k \leq t$,  there exists $ \bm{\alpha}_k:=(\alpha_{k}^1,\hdots,\alpha_{k}^n) \in \mathbb{Z}_q^n $ such that
\begin{align} \label{alpha}
f_{i}(\bm{z}_k):=\sum_{j=1}^n \alpha_{k}^j \cdot f_{i}(\bm{z}_j)~(mod~q)
\end{align}
Consequently, $\bm{c}_i^{\prime}= \tilde{\bm{c}}_i\cdot \bm{A} \in \mathbb{Q}^t$ (this multiplication is performed by considering the elements of $\bm{A}$ to be in $\mathbb{Q}$) where $\bm{A}$ is given by
\begin{align}
\bm{A}:=\left[
\begin{array}{c|c|c}
\bm{I}_n &~ \bm{0} &~ \begin{matrix}
                        \alpha_{\ell+1}^1  ~&~ \hdots ~&~ \alpha_{t}^1 \\
                        \vdots ~&~  \ddots ~&~ \vdots \\
                        \alpha_{\ell+1}^n  ~&~ \hdots ~&~ \alpha_t^n
                        \end{matrix}\\
\midrule
\bm{0} &~ \bm{I}_{\ell-n} &~ \bm{0}
\end{array}
\right] \in \mathbb{Z}_q^{\ell \times t}
\end{align}
%Hence, $\bm{c}_1^{\prime}=\bm{A} \cdot \tilde{\bm{c}}_1$ and $\bm{c}_2^{\prime}=\bm{A} \cdot \tilde{\bm{c}}_2$.

\vspace{1em}
\textbf{Step 3:} $\bm{c}_{mult}^{\prime}\in \mathbb{Q}^t$  in Equation (\ref{cmult_dash}) is obtained from $\bm{c}_1^{\prime}$ and $\bm{c}_2^{\prime}$ by taking their element-wise product and then multiplying the last $\ell-n$ elements by $2/q$.
%Take element-wise product of $\bm{c}_1^{\prime}$ and $\bm{c}_2^{\prime}$ and then multiply the entries from $n+1$ to $\ell$ with $2/q$ to get $\bm{c}_{mult}^{\prime}\in \mathbb{Q}^t$ given in equation (\ref{cmult_dash}). 
This operation can be done by evaluating a tensor $\mathbfcal{U} \in \mathbb{Q}^{t \times t\times t}$ on $\bm{c}_1^{\prime}$ and $\bm{c}_2^{\prime}$. For $1 \leq i \leq t$, let $\bm{U}_i$ denote the $i$-th frontal slice of $\mathbfcal{U}$.  For $1 \leq i \leq n$ and $\ell+1 \leq i \leq t$, the corresponding matrix $\bm{U}_i$, has  1 in the $(i,i)$-th position  and 0 everywhere else.  For $n+1 \leq i \leq \ell$,  $\bm{U}_i(i,i)=2/q$ and $\bm{U}_i(i,j)=0$ when $i \neq j$. Therefore,
\begin{align}
    \bm{c}_{mult}^\prime=\begin{bmatrix}
    \bm{c}_1^\prime \bm{U}_1 {\bm{c}_2^\prime}^{T} & ~\hdots~ & \bm{c}_1^\prime \bm{U}_t {\bm{c}_2^\prime}^T
    \end{bmatrix} \in \mathbb{Q}^t
\end{align}
%for $1 \leq k \leq n$, $\bm{P}_k$ is an $\ell \times \ell$ matrix with $\bm{P}_k(k,k)=1$ and $\bm{P}_k(i,j)=0$ for $1 \leq i,j \leq \ell,i =j \neq k$. For $n+1 \leq k \leq \ell$, $\bm{P}_k$ is an $\ell \times \ell$ matrix where $\bm{P}_k(k,k)=2/q$ with the rest of the entries being equal to 0. 

Given $ \tilde{\bm{c}}_1 $ and $ \tilde{\bm{c}}_2 $ from Step 1,  $ \bm{c}_{mult}^{\prime} $ is obtained by evaluating the following tensor $\mathbfcal{T} \in \mathbb{Q}^{\ell \times \ell \times t}$ on  $\tilde{\bm{c}}_1 $ and $ \tilde{\bm{c}}_2 $
\begin{align}
\mathbfcal{T}=\mathbfcal{U} \times_1 \bm{A} \times_2 \bm{A} \in \mathbb{Q}^{\ell \times \ell \times t}
\end{align}
For example, if $n=2$ and $\ell=4$, then the frontal slices of $\mathbfcal{T}$ for $1 \leq i \leq \ell$ can be represented by the following matrices: %$\bm{M}_k$ for $k=\ell+1$ is given by the matrix
{ \begin{align} \label{M1,M2}
\bm{T}_{1} =\begin{bmatrix}
1 & 0 & 0 & 0\\
0 & 0 & 0 & 0\\
0 & 0 & 0 & 0\\
0 & 0 & 0 & 0
\end{bmatrix},~
\bm{T}_{2} =\begin{bmatrix}
0 & 0 & 0 & 0\\
0 & 1 & 0 & 0\\
0 & 0 & 0 & 0\\
0 & 0 & 0 & 0
\end{bmatrix},~
\bm{T}_{3} =\begin{bmatrix}
0 & 0 & 0 & 0\\
0 & 0 & 0 & 0\\
0 & 0 & \frac{2}{q} & 0\\
0 & 0 & 0 & 0
\end{bmatrix},~
\bm{T}_{4} =\begin{bmatrix}
0 & 0 & 0 & 0\\
0 & 0 & 0 & 0\\
0 & 0 & 0 & 0\\
0 & 0 & 0 & \frac{2}{q}
\end{bmatrix} 
\end{align}}
 For $\ell+1 \leq i \leq t$, ${\bm{T}}_i $ is a matrix of the form
\begin{align}
\bm{T}_i:=\left[
\begin{array}{c|c}
\begin{matrix}
(\alpha_i^1)^2 ~&~ \alpha_i^1\alpha_i^2 ~&~ \hdots ~&~ \alpha_i^1\alpha_i^n \\
%\alpha_k^1\alpha_k^2 ~&~ (\alpha_k^2)^2 ~&~ \hdots ~&~ \alpha_k^2 \alpha_k^n \\
\vdots ~&~ \vdots ~&~ \ddots ~&~ \vdots \\
\alpha_i^1\alpha_i^n ~&~ \alpha_i^2\alpha_i^n ~&~ \hdots ~&~ (\alpha_i^n)^2
\end{matrix} &~ {0}^{n \times (\ell-n)} \\
\midrule
{0}^{(\ell-n) \times n} &~ {0}^{(\ell-n)\times(\ell-n)}
\end{array}
\right]\in \mathbb{Q}^{\ell \times \ell}
\end{align}
% We now multiply the tensor $\mathbfcal{P}$ with $\bm{D}$ and $\bm{D}^T$ to get
%\begin{align}
%\mathbfcal{T}=\mathbfcal{P} \times_1 \bm{D}^T \times_2 \bm{D}^T \in \mathbb{Q}^{\ell \times \ell \times t}
%\end{align}
%where $\times_i$ denotes the mode-$i$ product of tensors. 
Therefore, given $ \tilde{\bm{c}}_1 $ and $ \tilde{\bm{c}}_2 $ from Step 1, %${\bm{c}}_1=\bm{R}(\bm{m}_1+f_1(\bm{Z})+\bm{e}_1)~\text{mod}~q$, ${\bm{c}}_2=\bm{R}(\bm{m}_2+f_2(\bm{Z})+\bm{e}_2)~\text{mod}~q$ and the tensor ${\mathbfcal{T}}$, 
we can compute the vector $\bm{c}_{mult}^{\prime}$ as:
\begin{align} \label{c_mult1}
\bm{c}_{mult}^{\prime}:=\tilde{\bm{c}}_1 {\mathbfcal{T}} \tilde{\bm{c}}_2^T= \begin{bmatrix} \tilde{\bm{c}}_1 \bm{T}_1 \tilde{\bm{c}}_2^T &~ \cdots ~& \tilde{\bm{c}}_1 \bm{T}_t \tilde{\bm{c}}_2^T\end{bmatrix} \in \mathbb{Q}^t
\end{align}
%where $\bm{T}_i \in \mathbb{Q}^{\ell \times \ell}$ denotes the $i^{th}$ frontal slice of $\mathbfcal{T}$ for $1 \leq i \leq t$. 

\vspace{1em}

\textbf{Step 4:} %Observe that, for $n+1 \leq j \leq \ell$, 
% \begin{align}
%  \bm{c}_{mult}^{\prime}(j)= \frac{2}{q} \left(m_{1,j-n} \left\lfloor \frac{q}{2} \right\rfloor +{e}_{1,j-n}\!+qK_{1,j}\right)\left(m_{2,j-n} \left\lfloor \frac{q}{2} \right\rfloor + {e}_{2,j-n}+qK_{2,j}\right)
% \end{align}
The evaluations of polynomials in $\mathcal{I}_{\leq 2r}$ on $\bm{z}_1,\bm{z}_2,\hdots,\bm{z}_t$ constitute an $n_1$-dimensional subspace of $\mathbb{Z}_q^t$. Because of the way in which the points $(\bm{z}_1,\hdots,\bm{z}_n,\bm{z}_{\ell+1},\hdots,\bm{z}_t)$ are chosen, the evaluations of $f_1f_2 \in \mathcal{I}_{\leq 2r}$ at the points $(\bm{z}_{n+1}, \hdots,\bm{z}_\ell)$ can be written as a linear combination of its evaluations at $(\bm{z}_1,\hdots,\bm{z}_n,\bm{z}_{\ell+1},\hdots,\bm{z}_t)$. Therefore, for some $(\beta_{j}^1,\hdots,\beta_{j}^n,\beta_j^{\ell+1},\hdots,\beta_j^t) \in \mathbb{Z}_q^{n_1}$, $n+1 \leq j \leq \ell$, \begin{align}
f_1f_2(\bm{z}_j):=\sum_{i=1}^n \beta_{j}^i \cdot f_1f_2(\bm{z}_i)+\!\!\!\sum_{i=\ell+1}^t \beta_{j}^i \cdot f_1f_2(\bm{z}_i)~(mod~q)
\end{align}
 Let $\bm{B}_1$ be a matrix of size $n \times (\ell-n) $ such that $\bm{B}_1(i,j)=\beta_j^i$ for $1 \leq i \leq n$ and $n+1 \leq j \leq \ell$. Similarly, let $\bm{B}_2$ be a matrix of size $ (t-\ell)\times(\ell-n)$ such that $\bm{B}_2(i,j)=\beta_j^i$ for $n+1 \leq j \leq \ell$ and $\ell+1 \leq i \leq t$. Consider the following block matrix 
 \begin{align}
 \bm{B}=\left[
\begin{array}{c|c|c}
{\bm{I}}_n &~ \bm{B}_1 &~\bm{0} \\
\hline
\bm{0} &~ \bm{I}_{\ell-n} &~\bm{0} \\
\hline
\bm{0} &~ \bm{B}_2 &~ \bm{I}_{t-\ell}
\end{array}
\right] \in \mathbb{Z}_q^{t \times t}
 \end{align}
  $\bm{c}_{mult}^{\prime\prime}$ is obtained by multiplying $\bm{c}_{mult}^{\prime}$ with the matrix $\bm{B}$ i.e. $\bm{c}_{mult}^{\prime\prime}=\bm{c}_{mult}^{\prime} \cdot \bm{B}$ (considering the elements of $\bm{B}$ to be in $\mathbb{Q}$).
%Then, multiply $\bm{c}_{mult}^{\prime}$ by the matrix $\bm{E}$ to get $\bm{c}_{mult}^{\prime\prime}=\bm{E} \cdot\bm{c}_{mult}^{\prime} \in \mathbb{Q}^t$ such that equation (\ref{c_mult2}) is satisfied for $n+1 \leq j \leq \ell$.
%\begin{align} 
%\bm{c}_{mult}^{\prime\prime}(j) %&=\bm{E} \bm{c}_{mult}^{\prime}(j) \notag\\
%&=\frac{2}{q} \!\left(m_1 \left\lfloor \frac{q}{2} \right\rfloor \!+\!{e}_{1,j}\!+qK_{1,j}\!\right)\!\!\left(m_2 \left\lfloor \frac{q}{2} \right\rfloor \!+ {e}_{2,j}+qK_{2,j}\right)\!+\!f_1f_2(\bm{z}_j)
%\end{align}
%Consequently, the tensor can be modified as:
%\begin{align} \label{tensor T}
%\mathbfcal{U}=\mathbfcal{T}  \times_3 \bm{E} =\mathbfcal{P} \times_1 \bm{D}^T \times_2 \bm{D}^T\times_3 \bm{E}\in \mathbb{Q}^{\ell \times \ell \times t}
%\end{align}
%Therefore, $\bm{c}_{mult}^{\prime\prime}= \mathbfcal{U}(\bm{c}_1,\bm{c}_2) \in \mathbb{Q}^t$. 
%Thus, multiplying two vectors of size $\ell$, we get a vector of size $t$. We use the following technique to reduce the size of the resultant vector back to $\ell$.
\vspace{1em}

\textbf{Step 5:} In order to transform the vector $ \bm{c}_{mult}^{\prime\prime} $ generated in the above step to a valid ciphertext $\bm{c}_{mult}$ we first define a map from $\mathcal{I}_{\leq 2r}$ to $\mathcal{I}_{\leq r}$ as follows. 

Let LM$(\mathcal{I})$ be the set of leading monomials of elements of $\mathcal{I}$. Let LM$(\mathcal{I})_{r+1} := \{\mu_1,\mu_2,\ldots,\mu_M\}$ be the set of  all monomials of degree $r+1$ in LM$(\mathcal{I})$.  For each monomial $\mu_i$, choose a polynomial $g_i\in\mathcal{I}$ such that the leading term of $g_i$ is $\mu_i$ for $1 \leq i \leq M$. Let $\mathcal{G}:=\{ g_1,\hdots,g_{M} \}$ be the set of these polynomials. Given $f \in \mathcal{I}_{\leq 2r}$, serially divide $f$ by the set of polynomials $\mathcal{G}$ using the degree reverse lexicographic order (At each step divide the remainder obtained in the previous step by the next $g_i$). Let $f_{\mathcal{G}}$ be the final remainder obtained. Note that the map from $f$ to $f_{\mathcal{G}}$ is a linear one.

%In order to ensure that this process gives us a nonzero polynomial, $f=f_1f_2 \in \mathcal{I}_{\leq 2r}$ must not lie in the ideal generated by the polynomials $(g_1,\hdots,g_{M})$. A way of doing  this is explained in Section \ref{Ideal instantiation}. %The division is performed in a similar way to the division algorithm given in Theorem 3 of \cite{Cox1992}. 

The above linear map from $\mathcal{I}_{\leq 2r}$ to $\mathcal{I}_{\leq r}$ naturally gives rise to the following linear map $\mathcal{L}$ from the evaluations of polynomials in $ \mathcal{I}_{\leq 2r}$ at $(\bm{z}_1,\hdots,\bm{z}_t)$ to the evaluations of polynomials in $ \mathcal{I}_{\leq r}$ at $(\bm{z}_1,\hdots,\bm{z}_\ell)$. 
\begin{equation*}
    \mathcal{L} \left(f(\bm{z}_1),f(\bm{z}_2),\ldots,f(\bm{z}_t)\right) = \left(f_{\mathcal{G}}(\bm{z_1}),f_{\mathcal{G}}(\bm{z}_2),\ldots,f_{\mathcal{G}}(\bm{z}_\ell)\right)
\end{equation*}

Let $(f^1,\hdots,f^{n_1})$ be a basis for $\mathcal{I}_{\leq 2r}$ and let $(f_{\mathcal{G}}^1,\hdots,f_{\mathcal{G}}^{n_1})$ be the respective remainders after serially dividing by the elements of ${\mathcal{G}}$. Let $\bm{F}_1 \in \mathbb{Z}_q^{ n_1 \times t}$ and $\bm{F}_2\in \mathbb{Z}_q^{ n_1 \times \ell}$ be the following matrices,
\begin{align}
~~~\bm{F}_1:=\begin{bmatrix}
f^1(\bm{z}_1) & \hdots & f^1(\bm{z}_\ell) & \hdots & f^1(\bm{z}_t)  \\
\vdots        & \ddots & \vdots           & \ddots & \vdots\\
f^{n_1}(\bm{z}_1) & \hdots & f^{n_1}(\bm{z}_\ell) & \hdots & f^{n_1}(\bm{z}_t)\\
\end{bmatrix},
\end{align}
  
\begin{align}
\bm{F}_2:= \begin{bmatrix}
f_{\mathcal{G}}^1(\bm{z}_1) & \hdots & f_{\mathcal{G}}^1(\bm{z}_{\ell}) \\
\vdots & \ddots & \vdots \\
f_{\mathcal{G}}^{n_1}(\bm{z}_1) & \hdots & f_{\mathcal{G}}^{n_1}(\bm{z}_\ell)
\end{bmatrix}
\end{align}
Every $\bm{Q}$ that satisfies the equation $ \bm{F}_1 \cdot \bm{Q}= \bm{F}_2~({mod}~q)$, defines a map $\mathcal{L}_Q$ from $\mathbb{Z}_q^t$ to $\mathbb{Z}_q^{\ell} $ which when restricted to the evaluations of $\mathcal{I}_{\leq 2r}$ results in $\mathcal{L}$. In particular, there exist solutions $\bm{Q}$ that have the following structure,
%Then, $\bm{Q} \cdot \bm{F} = \bm{G}~({mod}~q)$, where each row vector of $\bm{Q}$ can be determined by solving a system of $n_1$ equations in $t$ unknowns over $\mathbb{Z}_q$. Therefore, $\bm{Q}$ can take the form
\begin{align}
\bm{Q}=\left[
\begin{array}{c|c|c}
\bm{Q}_1^{n \times n} &~ {0}^{n\times(\ell-n)} ~& \bm{Q}_2^{n \times(t-\ell)} \\
\midrule
\bm{Q}_3^{(\ell-n)\times n} & \bm{I}_{\ell-n} & \bm{Q}_4^{(\ell-n)\times(t-\ell)}
\end{array}
\right]^T \in \mathbb{Z}_q^{t \times \ell}
\end{align}
%where $ {0}^{n\times(\ell-n)} $ denotes the all zero matrix of size $  n\times(\ell-n)$ and $\bm{I}_{\ell-n}  $ denotes the identity matrix of size $ \ell-n $. 

Now, when $\bm{c}_{mult}^{\prime\prime}$ is multiplied  with the matrix $\bm{Q}$ (considering the elements of $\bm{Q}$ to be in $\mathbb{Q}$) we get the following vector $ \tilde{\bm{c}}_{mult}$
\begin{equation}\label{cmult_tilde}
\resizebox{.922\hsize}{!}{
$\tilde{\bm{c}}_{mult} 
\!=\!\!\begin{bmatrix}
f_{mult}(\bm{z}_1)+qK_1^\prime \\
\vdots \\
f_{mult}(\bm{z}_n)+qK_n^\prime\\
\!\frac{2}{q}\! \left(m_{1,n+1} \!\left\lfloor \frac{q}{2} \right\rfloor \!+\!{e}_{1,n+1}^\prime\!+\!qK_{1,n+1}\right)\!\!\left(m_{2,n+1}\! \left\lfloor \frac{q}{2} \right\rfloor \!+\! {e}_{2,n+1}^\prime\!+\!qK_{2,n+1}\!\right)\!+\!f_{mult}(\bm{z}_{n+1})+qK_{n+1}^\prime\! \\
\vdots \\
\frac{2}{q}  \left(m_{1,\ell} \left\lfloor \frac{q}{2} \right\rfloor +{e}_{1,\ell}^\prime+qK_{1,\ell}\right)\!\left(m_{2,\ell} \left\lfloor \frac{q}{2} \right\rfloor + {e}_{2,\ell}^\prime+qK_{2,\ell}\right)\!+\!f_{mult}(\bm{z}_{\ell})+qK_{\ell}^\prime\!
\end{bmatrix}^T $}
\end{equation}
for some $K_j^\prime \in \mathbb{Z}$ for $1 \leq j \leq \ell$. Here $f_{mult}$ denotes the  remainder obtained after serially dividing $f_1f_2$ by the elements of $\mathcal{G}$.
 Multiplying $ \tilde{\bm{c}}_{mult} $ by $\bm{R}$ gives us a vector in $\mathbb{Q}^\ell$ which when rounded to the closest integer vector is equivalent $ mod~q$ to a vector which gives $\bm{m}_1 \odot \bm{m}_2$ on decryption i.e,
\begin{align}\label{cmult}
\bm{c}_{mult}=\left\lfloor \tilde{\bm{c}}_{mult} \cdot \bm{R} \right\rfloor ~{mod}~q \in \mathbb{Z}_q^\ell
\end{align}

The sequence of steps that transform the pair $\bm{c}_1,\bm{c}_2$ to $ \tilde{\bm{c}}_{mult} \bm{R}$ constitute a bilinear map from $\mathbb{Q}^\ell \times \mathbb{Q}^\ell$ to $\mathbb{Q}^\ell$. This map, denoted by $\mathcal{B}_{\mathbfcal{M}}$, can be represented by a 3-way tensor $\mathbfcal{M}$ where $\mathbfcal{M}$ is given by
\begin{align} \label{M2}
{\mathbfcal{M}} = \mathbfcal{T} \times_1 \bm{D}_1 \times_2 \bm{D}_2 \times_3 (\bm{R}^T\bm{Q}^T\bm{B}^T) \in \mathbb{Q}^{\ell \times \ell \times \ell}
\end{align}
The tensor $ \mathbfcal{M} $ is the evaluation key for multiplication. If $\bm{M}_1,\hdots,\bm{M}_\ell$ denote the frontal slices of $\mathbfcal{M}$, then using $\mathbfcal{M}$, $\bm{c}_{mult}$ can be computed as:
\begin{align} \label{c_mult}
\bm{c}_{mult}=\left\lfloor \mathcal{B}_{\mathbfcal{M}} (\bm{c}_1 ,{\bm{c}}_2)\right\rfloor\,{mod}\,q=\!\begin{bmatrix}\left\lfloor\bm{c}_1 {\bm{M}_1} {\bm{c}}_2^T \right\rfloor &~ \cdots ~&\left\lfloor\bm{c}_1 {\bm{M}_\ell} {\bm{c}}_2^T \right\rfloor\end{bmatrix}\,{mod}\,q\in \mathbb{Z}_q^\ell
\end{align}

%One of the fastest method of solving a polynomial system is by computing its Gr\"{o}bner basis. Some of the best known Gr\"{o}bner bases computation algorithms are \cite{faugere1993efficient,faugere1999new,faugere2002new}. The complexity of computing a Gr\"{o}bner basis is extensively studied in \cite{lazard1983grobner,lakshman1991single,dickenstein1991membership,bardet2004complexity,bardet2015complexity}. For $n,\ell=\mathcal{O}(\lambda)$, the complexity of computing a Gr\"{o}bner basis for the given system of equations is exponential in the security parameter $\lambda$. 

%One such method is the hybrid approach proposed in \cite{bettale2009hybrid}. The asymptotic complexity of the hybrid approach is determined to be $ 2^{(3.31-3.62\,\text{log}\,(q)^{-1})N_v} $ \cite{bettale2012solving}. For $\ell=\mathcal{O}(n)$ and $q=\mathcal{O}(n^2)$, the complexity of solving the PoSSo problem in the proposed case using the hybrid approach is of the order of $2^{\tilde{\mathcal{O}}(n^2)}$.

\subsubsection*{Correctness}
% If $f_{mult}=f_1f_2~mod~\mathcal{G}$, then $\bm{c}_{mult}$ from equation (\ref{cmult}) and (\ref{c_mult}) can be written as
% \begin{align} \label{c_mul}
% \bm{c}_{mult}= \!\!\begin{bmatrix}
%  f_{mult}(\bm{z}_1) \\
%  \vdots \\
%  f_{mult}(\bm{z}_n) \\
%  (m_{1,n+1}m_{2,n+1})\cdot \left\lfloor \frac{q}{2} \right\rfloor+ f_{mult}(\bm{z}_{n+1})+{e}_{mult,n+1} \\
%  \vdots \\
%  (m_{1,\ell}m_{2,\ell})\cdot \left\lfloor \frac{q}{2} \right\rfloor+ f_{mult}(\bm{z}_{\ell})+ {e}_{mult,\ell} 
%  \end{bmatrix}^T \! \bm{R}_1 ~(mod~q)\in \mathbb{Z}_q^\ell
% \end{align}
%where $e_{mult,j}$ denotes the noise in each of the corrupted entries after multiplication for $n+1 \leq j \leq \ell$. 
If $f_{mult}(\bm{Z})=\bm{y}_{mult}\cdot \bm{S}_{\text{enc}}$ for some $\bm{y}_{mult} \in \mathbb{Z}_q^n$ and $\bm{e}_{mult}=(0,\hdots,0,e_{mult,n+1},\hdots,e_{mult,\ell})$ denotes the noise vector, % and $\bm{m}_{mult}=(\bm{m}_1\odot \bm{m}_2)\left\lfloor{q}/{2} \right\rfloor$, 
then the resultant ciphertext after the multiplication process can be written as
\begin{align}
\bm{c}_{mult}= \left((\bm{p}_1\odot \bm{p}_2)\left\lfloor\frac{q}{2} \right\rfloor+ \bm{y}_{mult}\cdot \bm{S}_{\text{enc}}+{\bm{e}}_{mult} \right) \cdot \bm{R}~{mod}~q \in \mathbb{Z}_q^\ell
\end{align}
 Then $\bm{c}_{mult}\cdot\bm{S}_{\text{dec}}=(\bm{m}_1\odot \bm{m}_2)\left\lfloor\frac{q}{2} \right\rfloor+\tilde{\bm{e}}_{mult}~({mod}~q) $ where $\tilde{\bm{e}}_{mult}=(e_{mult,n+1},\hdots,e_{mult,\ell})$. The decryption process outputs $ \bm{m}_1\odot \bm{m}_2$ if $ \norm{{\bm{e}}_{mult}}_{\infty} < \left\lfloor {q}/{2} \right\rfloor/2$.
 
\subsubsection*{Security with multiplicative homomorphism}

The entries of the tensor $\mathbfcal{M} \in \mathbb{Q}^{\ell \times \ell\times \ell}$ are polynomials in the entries of the matrices $\bm{S},\bm{R}_1,\bm{R}_2,\bm{Q}$, the $\bm{\epsilon}_i$s, the $\alpha_i^j$s and $\beta_i^j$s. Equating these polynomials with a given instance of the evaluation key results in a system of $\mathcal{O}(\ell^3)$  equations in $\mathcal{O}(\ell)^2$ variables.
%(The entries of the tensor $\mathbfcal{M} \in \mathbb{Q}^{\ell \times \ell\times \ell}$ forms a system of $\frac{1}{2}(\ell^3+\ell^2)$ equations in $(4n\ell+2t\ell-3n^2-2\ell^2)$ 
 The $\alpha_i^j$s, $\beta_i^j$s and the entries of the matrix $\bm{Q}$ depend on the extra $t-\ell$ points that are chosen independent of the secret key. The entries of $\bm{Q}$ also depend on the polynomials chosen for the quotienting operation.
Therefore, in order to retrieve the secret key from the evaluation key one would have to solve a  system of polynomial equations. The problem of Polynomial System Solving (PoSSo) is known to be NP-hard in general.   Most multivariate public key schemes rely on the hardness of solving this problem. For detailed analysis of this problem, one may refer to \cite{lazard1983grobner,kipnis1999cryptanalysis,faugere2003algebraic,courtois2002cryptanalysis,albrecht2011polly,bettale2009hybrid}. Further, the system of equations in this case is underdetermined (as a system of equations over $\mathbb{Q}$).

Now, the multiplication operation defines an `almost' bilinear map from $\mathbb{Z}_q^\ell \times \mathbb{Z}_q^\ell$ to $\mathbb{Z}_q^\ell$ (The rounding operations induce some non-linearity). In the absence of the $\bm{\epsilon}_i$s, the noiseless encryptions of $0$ constitute an invariant subspace of this map. This could potentially reveal information about the secret key (Although, to the best of the authors' knowledge there are no efficient algorithms to extract such subspaces for all $\ell$). The $\bm{\epsilon}_i$s ensure that this invariance is removed.

\subsubsection*{Noise in Multiplication}
Let us now analyze the noise in the decryption of $\bm{c}_{mult}$. Observe that, for $n+1 \leq j \leq \ell$, the noise in $\bm{c}_{mult}(j)$ is same as that in $\tilde{\bm{c}}_{mult}(j)$. If $e_{mult,j}$ denotes the noise in the $j^{th}$ entry of $\bm{c}_{mult}$ then using Equation (\ref{cmult_tilde}), we get %for $n+1 \leq j \leq \ell$,  %Since, $\bm{s}^{\scriptscriptstyle \text{HSM}}=\tilde{\bm{s}}_j$, we have 

% {\small \begin{align} \label{mult_error}
% \tilde{\bm{c}}_{mult}(j) &=\frac{2}{q} \left(m_{1,j} \left\lfloor \frac{q}{2} \right\rfloor +{e}_{1,j}\!+qK_{1,j}\right)\left(m_{2,j} \left\lfloor \frac{q}{2} \right\rfloor + {e}_{2,j}+qK_{2,j}\right)+f_{mult}(\bm{z}_j)\notag\\
% &= m_{1,j}m_{2,j}\left\lfloor \frac{q}{2} \right\rfloor+f_{mult}(\bm{z}_j) +\frac{q-1}{q}(m_{1,j}e_{2,j}+m_{2,j}e_{1,j})+(2e_{1,j}-m_{1,j})K_{2,j} \notag\\ 
% &\quad+(2e_{2,j}-m_{2,j})K_{1,j}-\frac{m_{1,j}m_{2,j}}{2q}+\frac{2}{q}e_{1,j}e_{2,j} +q(m_{1,j}K_{2,j}+m_{2,j}K_{1,j}+2K_{1,j}K_{2,j})\notag\\
% &= m_{1,j}m_{2,j}\left\lfloor \frac{q}{2} \right\rfloor +f_{mult}(\bm{z}_j)+ e_{mult,j} +q(m_{1,j}K_{2,j}+m_{2,j}K_{1,j}+2K_{1,j}K_{2,j})
% \end{align}}
\begin{align}
 e_{mult,j}&= \frac{q-1}{q}(m_{1,j}e_{2,j}^\prime+m_{2,j}e_{1,j}^\prime)+(2e_{1,j}^\prime-m_{1,j})K_{2,j} +(2e_{2,j}^\prime-m_{2,j})K_{1,j}\notag\\ &\quad -\frac{m_{1,j}m_{2,j}}{2q}+\frac{2}{q}e_{1,j}^\prime e_{2,j}^\prime
 \end{align}%Note that, $\bm{S}_{dec}(:,i)=\tilde{\bm{s}}_{n+i}$ for $1 \leq i \leq \ell-n$. 
% Therefore, for $n+1 \leq j \leq \ell$, 
% %{ \begin{align} 
% $\langle \bm{c}_{mult},\tilde{\bm{s}}_j \rangle  = m_{1,j}m_{2,j}\left\lfloor \frac{q}{2} \right\rfloor + e_{mult,j} ~({mod}\,q)$.
% %\end{align}}
% Hence, 
% \begin{align}
%  \bm{c}_{mult}\cdot\bm{S}_{d}=(\bm{m}_1 \odot \bm{m}_2) \cdot \left\lfloor \frac{q}{2} \right\rfloor+\tilde{\bm{e}}_{mult}~({mod}\,q)
% \end{align}
% where $ \tilde{\bm{e}}_{mult}=(e_{mult,n+1},\hdots,e_{mult,\ell}) $ is the noise in multiplication. 
%Observe that, the most significant contribution to the  error is  from the term $(2e_{1,j}-m_1)K_{2,j}+ (2e_{2,j}-m_2)K_{1,j} $. 

The most significant term in $e_{mult,j}$ is $ (2e_{1,j}^\prime-m_{1,j})K_{2,j} +(2e_{2,j}^\prime-m_{2,j})K_{1,j} $ where the $K_{i,j}$s are generated due to the multiplication of $\bm{c}_i$s with the matrix $\bm{D}_i$ in Step 1. Observe that the $j^{th}$ entry of $\bm{c}_i\bm{D}_i$ is equal to $\langle \bm{c}_i\bm{R}^{-1},\tilde{\bm{s}}_j \rangle + \langle \bm{c}_i,(\bm{\epsilon}_i(:,j),\bm{0}) \rangle$ where $\bm{0}$ is the zero vector in $\mathbb{Q}^{\ell-n}$ and $\langle \bm{c}_i\bm{R}^{-1}, \tilde{\bm{s}}_j \rangle =m_{i,j} \left\lfloor \frac{q}{2} \right\rfloor +{e}_{i,j}+qK_{i,j}$. The magnitude of %Given that, $\langle \bm{c}_1, \tilde{\bm{s}}_j \rangle =m_{1,j} \left\lfloor \frac{q}{2} \right\rfloor +{e}_{1,j}+qK_{1,j}$ and $\langle \bm{c}_2 ,\tilde{\bm{s}}_j \rangle =m_{2,j} \left\lfloor \frac{q}{2} \right\rfloor +{e}_{2,j}+qK_{2,j}$, 
${K_{i,j}}$ is bounded by the one norm of $ \bm{R}^{-1}\tilde{\bm{s}}_j$ which is $\mathcal{O}(nq)$ since $\ell=\mathcal{O}(n)$. %Since $e_{i,j}$ and $ \langle \bm{c}_i,(\bm{E}_i(:,j),\bm{0} \rangle $ are bounded by $B$, $\abs{e_{mult,j}}$ is bounded by $\mathcal{O}()$. 
One could choose the secret key in such a way that the one norms of the $\bm{R}^{-1}\tilde{\bm{s}}_j$s are small. Alternatively, one could use a slightly modified version of the vector decomposition techniques given in \cite{brakerski2014leveled} to limit the values of the $\abs{K_{i,j}}$s. Firstly, for a suitable value of $u$ (which is $\mathcal{O}(1))$, we choose the entries of the $\bm{\epsilon}_i$s (all of which are less than 1) such that their binary expressions have less than $u$ bits  i.e. these entries can be written as $\sum_{k = 1}^ub_k2^{-k}$ for some $b_k$s in $ \{0,1\}$.  This technique consists of the following two functions
%Given that, $\langle \bm{c}_1, \tilde{\bm{s}}_j \rangle =m_{1,j} \left\lfloor \frac{q}{2} \right\rfloor +{e}_{1,j}+qK_{1,j}$ and $\langle \bm{c}_2 ,\tilde{\bm{s}}_j \rangle =m_{2,j} \left\lfloor \frac{q}{2} \right\rfloor +{e}_{2,j}+qK_{2,j}$, 
%$$\abs{K_{1,j}}$ and $\abs{K_{2,j}}$ are bounded by $ \ell_1(\tilde{\bm{s}}_j)$ which is $\bm{O}(nq)$. One could chose the secret key in such a way that the one norms of the $\bm{s}_j$s are small. Alternatively, one could use a slightly modified version of the vector decomposition techniques given in \cite{brakerski2014leveled} to limit the values of the $\abs{K_{i,j}}$s. Firstly, for a suitable value of $k$ (which is $\bm{O}(1))$, we choose the entries of the $\bm{\epsilon}_i$s (all of which are less than 1) such that their binary expressions have less than $k$ bits  i.e. these entries can be written as $\sum_{r = 1}^kb_r2^{-r}$ for some $b_r$s in $ \{0,1\}$.  This technique consists of the following two functions
\begin{itemize}
\item \textsf{BitDecomp}$_{q,u}(\bm{v})$: Given $\bm{v} \in \mathbb{Q}^\ell$, let $\bm{x}_i \in \{0,1\}^\ell$ be such that $\bm{v}=\sum_{i=-u}^{\lfloor \text{log}\,q\rfloor} 2^i \cdot \bm{x}_i~({mod}~q)$. Output the vector
%\vspace{-1em}
\begin{align*}
\left(\bm{x}_{-u},\ldots,\bm{x}_0,\ldots,\bm{x}_{\lfloor \text{log}\,q\rfloor} \right) \in \{0,1\}^{\ell(u+\lceil \text{log}\,q\rceil)}
\end{align*}
%\vspace{0.75ex}
\item \textsf{PowersOfTwo}$_{q,u}(\bm{w})$: Given $\bm{w} \in \mathbb{Z}^\ell$, output the vector
%\vspace{-1em}
\begin{align*}
\left( 2^{-u} \cdot\bm{w},\hdots,2^{-1}\cdot\bm{w},\bm{w},2 \cdot \bm{w},\hdots,2^{\lfloor \text{log}\,q\rfloor}\cdot \bm{w} \right)~{mod}~q \in \mathbb{Q}^{\ell(u+\lceil \text{log}\,q\rceil)}
\end{align*}
\end{itemize}
It can be easily verified that $\langle \bm{v},\bm{w}\rangle=\left\langle \textsf{BitDecomp}_{q,u}(\bm{v}),\textsf{PowersOfTwo}_{q,u}(\bm{w}) \right\rangle~{mod}~q$.

%Using these techniques, 
%Observe that, for $n+1 \leq j \leq \ell$ and $i\in \{1,2\}$, the $j$-th entry of  $\bm{c}_i\bm{D}_i$ is equal to  $\langle c_i,\tilde{\bm{s}}_j \rangle + \langle c_i,(\bm{\epsilon}_i(:,j),\bm{0} \rangle$ where $\bm{0}$ is the zero vector in $\mathbb{Q}^{\ell-n}$.
Let $\widetilde{\bm{D}}_i$s be the matrices got by applying \textsf{BitDecomp}$_{q,u}$ on the columns of $\bm{D}_i$s. Now, instead of multiplying the $\bm{c}_i$s with the $\bm{D}_i$s, if we multiply the \textsf{PowersOfTwo}$_{q,u}(\bm{c}_i)$s with the respective $\widetilde{\bm{D}}_i$s, the corresponding 
$K_{i,j}$s are given as
\begin{equation}
   {K_{i,j}} = \frac{1}{q}\cdot{\langle \textsf{PowersOfTwo}_{q,u}(\bm{c}_i), \textsf{BitDecomp}_{q,u}({\bm{D}_i}(:,j)) \rangle-m_{i,j}\lfloor q/2 \rfloor-e_{i,j}^{\prime}}
\end{equation}
Therefore, their magnitudes can be computed as
{ \begin{align} \label{K}
\abs{K_{i,j}} &= \frac{1}{q}\cdot\abs{\left\langle \textsf{PowersOfTwo}_{q,u}(\bm{c}_i), \textsf{BitDecomp}_{q,u}({\bm{D}_i}(:,j))  \right\rangle-m_{i,j}\lfloor q/2 \rfloor-e_{i,j}^{\prime}} \notag \\
 &\leq \frac{\abs{\left\langle \textsf{PowersOfTwo}_{q,u}(\bm{c}_i),\widetilde{\bm{D}}_i(:,j)  \right\rangle}}{q}+1 \notag\\
&\leq \frac{1}{2}\cdot\norm{\widetilde{\bm{D}}_i(:,j)}_1+1 \notag \\
 &\leq \frac{1}{2}\cdot \left(\ell(u+\lceil{\text{log}\,q})\rceil\right))+1 \notag \\
&= \mathcal{O}(n\,{\text{log}\,q}) \textrm{  ( Since }\ell \textrm{ is } \mathcal{O}(n) \textrm{ and } u \textrm{ is } \mathcal{O}(1))
\end{align}}

To accommodate these changes in the evaluation key, the tensor $\mathbfcal{M}$ can be modified as:
\begin{align}
\mathbfcal{M}=\mathbfcal{T} \times_1 \widetilde{\bm{D}} \times_2 \widetilde{\bm{D}} \times_3 (\bm{R}^T\bm{Q}^T\bm{B}^T)  \in \mathbb{Q}^{\ell(u+\lceil\text{log}\,q\rceil) \times \ell(u+\lceil\text{log}\,q\rceil)  \times \ell}
\end{align}
%Using the corresponding bilinear map $\mathcal{B}_{\mathbfcal{M}}: \mathbb{Q}^{\ell(u+\lceil\text{log}\,q\rceil)} \times \mathbb{Q}^{\ell(u+\lceil\text{log}\,q\rceil)} \rightarrow \mathbb{Q}^\ell$, 
Then, given $\bm{c}_1$ and $\bm{c}_2$, $\bm{c}_{mult}$ can be computed as
\begin{align}
\bm{c}_{mult}=\left\lfloor  \mathcal{B}_{\mathbfcal{M}}\left(\textsf{PowersOfTwo}_{q,u}({\bm{c}}_1) ,\textsf{PowersOfTwo}_{q,u}( {\bm{c}}_2)\right)  \right\rfloor ~{mod}~q\in \mathbb{Z}_q^\ell
\end{align}

\subsubsection*{Noise Magnitude}

%From Equation (\ref{mult_error}), noise after multiplication is given by $\tilde{\bm{e}}_{mult}=(e_{mult,n+1},\hdots,e_{mult,\ell})$ where $e_{mult,j}=\frac{q-1}{q}(m_{1,j}e_{2,j}+m_{2,j}e_{1,j})+(2e_{1,j}-m_{1,j})K_{2,j} +(2e_{2,j}-m_{2,j})K_{1,j}-\frac{m_{1,j}m_{2,j}}{2q}+\frac{2}{q}e_{1,j}e_{2,j}$ for $n+1 \leq j \leq \ell$. 
If the above mentioned vector decomposition techniques are used then $\abs{K_{i,j}} \leq \mathcal{O}(n\,\text{log}\,q)$  for $i \in \{1,2\}$. Since, $\abs{{e}_{i,j}^\prime} \leq 2B$ for $i \in \{1,2\}$ and $B <\left\lfloor {q}/{2} \right\rfloor/2 \leq  {q}/{4} $, the magnitude of the error after multiplication is as follows:
\begin{align} \label{mult noise}
\norm{{\bm{e}}_{mult}}_{\infty} \leq 4B+2(4B+1)\cdot\mathcal{O}(n\,\text{log}\,q)+\frac{8B^2+1}{q}=\mathcal{O}(n\,\text{log}\,q)\cdot B
\end{align}

\begin{theorem}
The proposed scheme with parameters $n,q,L,\mathcal{X}$ with $\abs{\mathcal{X}} \leq B$ can evaluate circuits of depth $L$ when $
q/B \geq \left(\mathcal{O}(n\,\text{log}\,q)\right)^L$.
\end{theorem}
\begin{proof}
From Lemma \ref{CoE}, noise in a fresh encryption is at most $B$. After one level of multiplication, it increases to $\mathcal{O}(n\,\text{log}\,q)\cdot B$. If $e_{mult}^i$ denotes the noise at level $i$, then $ \abs{e_{mult}^i} = \mathcal{O}(n\,\text{log}\,q) \cdot \abs{e_{mult}^{i-1}}$. Therefore, $\abs{e_{mult}^L} =\left(\mathcal{O}(n\,\text{log}\,q)\right)^L \cdot B$. For correctness of decryption, we need $ \abs{e_{mult}^L} \leq q/4 $. Hence, $q/B \geq \left(\mathcal{O}(n\,\text{log}\,q)\right)^L$.
\qed
\end{proof}

\subsubsection*{Parameters and Performance}

Similar to \cite{gentry2013homomorphic}, we choose $n=\mathcal{O}(\lambda)$ to be a fixed parameter and $\ell= \mathcal{O}(n)$ to be slightly bigger than $n$ such that $\ell-n=\mathcal{O}(1)$.
The proposed scheme can evaluate a circuit of depth $L$ as long as $q/B \geq \left(\mathcal{O}(n\,\text{log}\,q)\right)^L$. Therefore, we can choose $q$ to be of bit size $\mathcal{O}(L\,\text{log}\,n)$ similar to \cite{brakerski2014leveled,brakerski2012fully}. Gentry's bootstrapping theorem \cite{gentry2009fully} states that if the decryption circuit complexity of an $L$-homomorphic scheme is less than $L$, then there exists a leveled fully homomorphic encryption scheme. Decryption circuit complexity in the proposed scheme can be bounded by $ \mathcal{O}(\text{log}\,n+\text{log}\,\text{log}\,q) $ using similar techniques as in \cite{brakerski2014efficient}. For $L=\mathcal{O}(\text{log}\,n)$, $q/B$ in the proposed scheme is quasi-polynomial in $n$ and its security is based on the hardness of LWE for quasi-polynomial factors given by $\gamma = n^{\mathcal{O}(\text{log}\,n)}$ (since $\gamma=(q/B)\cdot \tilde{\mathcal{O}}(n) $).

The cost of multiplying two ciphertexts is of the order of $\mathcal{O}(\ell^3\,\text{log}^2\,q) = \tilde{\mathcal{O}}(n^3\cdot L^2)$ while that of adding two ciphertexts is $\mathcal{O}(\ell)$. Therefore, the per gate computation of the leveled FHE scheme is $\tilde{\mathcal{O}}(n^3\cdot L^2)$. %We give a comparison of the per gate computation of the proposed scheme with previous schemes both in LWE as well as RLWE instantiations in Table \ref{para}.
%\begin{table}[ht] 
%\setlength{\tabcolsep}{40pt}
%\centering
%\begin{tabularx}{\textwidth}[t]{cc}
%\hline
%Scheme & Per gate computation \\
%\hline
%\cite{brakerski2014leveled,brakerski2012fully} & $\tilde{\mathcal{O}}(\lambda^3 L^5)$(LWE) and $\tilde{\mathcal{O}}(\lambda L^3)$(RLWE)\\
%\cite{gentry2013homomorphic} & $\tilde{\mathcal{O}}\left((\lambda L)^{2.37}\right)$\\
%Proposed scheme & $ \tilde{\mathcal{O}}(\lambda^3) $\\
%\hline
%\end{tabularx}
%\caption{Per gate computation overhead of LWE-based schemes} \label{para}
%\end{table}%

\begin{comment}
\subsection{Generating the polynomials $(g_1,\hdots,g_{M})$} \label{Ideal instantiation}

We now show how the polynomials $(g_1,\hdots,g_{M})$ in Step 5 of section \ref{pk} are generated.
Observe that, if $f_1f_2 \in \langle g_1,\hdots,g_{M}\rangle$, then $f_1f_2~{mod}~\mathcal{G}$ could potentially go to $0$. To avoid such a situation, the generating polynomial $g$ is chosen to be a monic polynomial with a non-zero constant term.  Let $(p_1,\hdots,p_{M})$ be the set of all monomials of degree $(r+1-r_g)$. Choose each $p^{\prime}_i$ to be a polynomial with a non-zero constant term whose maximum degree monomial is the corresponding $p_i$. Then, $g_i=g \cdot p^{\prime}_i$ for $1 \leq i \leq M$. If we restrict the choice of polynomials during encryption to those with non zero constant terms, then $f_1f_2 \notin \langle g_1,\hdots,g_{M} \rangle$. Hence $f_1f_2~{mod}~\mathcal{G}$ will never be zero.
\end{comment}

\section{Private key to Public key Conversion} \label{PK}

The proposed scheme can be converted to a public key scheme as follows.  
For some $\epsilon > 0$, let $\bm{C}_0$ be a list of $d=(1+\epsilon)(\ell\,\text{log}\,q)$ encryptions of the zero vector under the private key scheme explained earlier in the paper. Let $\bm{b}_1,\hdots,\bm{b}_{\ell-n}$ be the standard basis for $\mathbb{Z}_q^{\ell-n}$, i.e., $\bm{b}_1=(1,0,0,\hdots,0), \bm{b}_2=(0,1,0,\hdots,0)$ and so on. Let $\bm{c}_{\bm{b}_1},\hdots,\bm{c}_{\bm{b}_{\ell-n}}$ be the encryptions of these vectors using the private key scheme. Construct a matrix $\bm{C}_{pk} \in \mathbb{Z}_q^{(\ell-n) \times \ell}$ by assigning $\bm{C}_{pk}(i,:)=\bm{c}_{\bm{b}_i}$ for $1 \leq i \leq \ell-n$. Then, the public key is given by $pk=(\bm{C}_0,\bm{C}_{pk})$ and the secret key is same as that of the private key scheme. 
The public key scheme can be described in terms of the following algorithms.
\begin{itemize}
\setlength\itemsep{1ex}
\item[$  \bullet$]{\textsf{PK.KeyGen}$(1^{\lambda})$}: It takes the security parameter $\lambda$ and outputs the public encryption key $pk=(\bm{C}_0,\bm{C}_{pk})$ and the secret decryption key $sk={\bm{S}}_{\text{dec}} $ where $\bm{S}_{\text{dec}}=\bm{R}^{-1}\cdot\begin{bmatrix} \bm{S} ~& \bm{I}_{\ell- n}\end{bmatrix}^T$.

\item[$  \bullet$]{\textsf{PK.Encrypt}$(pk,\bm{m})$}: To encrypt a message $\bm{m}\in\{0,1\}^{\ell-n}$, select a random subset $S $ of $\bm{C}_0$ and compute the ciphertext as
\begin{align}
\bm{c}=\bm{m} \cdot \bm{C}_{pk}+\sum_{\bm{c}_i \in S}\bm{c}_i~(mod~q)
\end{align}

\item[$  \bullet$]\textsf{PK.Decrypt}$(sk,\bm{c})$: Decryption is performed by computing  
\begin{align}
   \bm{m}=\left\lfloor \frac{1}{\left\lfloor q/2 \right\rfloor}\left( \bm{c}\cdot \bm{S}_{\text{dec}} \,{mod}\,q\right)\right\rceil\,{mod}~2
\end{align}
\end{itemize}
If $\tilde{\bm{e}}_{pk}=(\bm{0},\bm{e}_{pk})$ denotes the noise associated with the ciphertext $\bm{c}$, then  
{ \begin{align}
 \bm{c}\cdot\bm{S}_{\text{dec}}  =    \left(\bm{m} \cdot \bm{C}_{pk}+\sum_{\bm{c}_i \in S}\bm{c}_i\right)\cdot\bm{S}_{\text{dec}}  
  %=(\bm{m}+\bm{0}) \cdot \left\lfloor \frac{q}{2} \right\rfloor + \bm{e}_{pk}~({mod}\,q)\notag \\
  = \bm{m} \cdot \left\lfloor \frac{q}{2} \right\rfloor +\bm{e}_{pk}~({mod}\,q)
\end{align}}

Observe that $\norm{\bm{e}_{pk}}_{\infty} \leq (wt(\bm{m})+\abs{S})\cdot B$ where $wt(\bm{m})$ denotes the weight of the vector $\bm{m}$ and $\abs{S}$ denotes the cardinality of the set $S$. Therefore, the decryption function outputs $\bm{m}$ when $(wt(\bm{m})+\abs{S})\cdot B < q/4$. %The noise distribution of the entries of $\bm{e}_{pk}$ is a discrete Gaussian with standard deviation $\sqrt{\abs{S}}(\alpha q)$. 
The security of the scheme follows  from the security of the private key scheme and Claim 5.3 in \cite{regev2009lattices} (a special case of the leftover hash lemma). This claim is restated as follows
\begin{lemma}{\bf[Claim 5.3 \cite{regev2009lattices}]}\label{leftoverhash}
Let $S=\{\bm{g}_1,\hdots,\bm{g}_d\}$ be some subset of $\mathbb{Z}_q^\ell$ for some $d \in \mathbb{N}$. Then, for a uniform choice of $S$, given a hash function $h_{\bm{A}}:\{0,1\}^d \rightarrow \mathbb{Z}_q^\ell$ defined as $h_{\bm{A}}(\bm{x})=\bm{x}\bm{A}^T~mod~q$ where $\bm{A}(:,j)=\bm{g}_j$ for $1 \leq j \leq d$, the expectation of the statistical distance of the distribution on $\bm{x}\bm{A}^T~mod~q$ from uniform has an upper bound of$\sqrt{q^\ell/2^d}$.  Further, the probability of this statistical distance being greater than $\sqrt[4]{q^\ell/2^d}$ is upper bounded by $\sqrt[4]{q^\ell/2^d}$.
 \end{lemma}
Now, for some $\epsilon>0$, if $d=(1+\epsilon)\ell\,\text{log}\,q$, then  $\sqrt{q^\ell/2^d}$ and $\sqrt[4]{q^\ell/2^d}$ are  negligible in $\ell$.

\begin{lemma}
For parameters $n,\ell,d$ and $\mathcal{X}$, if there exists an efficient algorithm  that can distinguish between encryptions of any two distinct messages $\bm{m}_1$ and $\bm{m}_2$ under the above described public key scheme, then there exists a message $\bm{m}$ and an algorithm  that can distinguish between encryptions of $\bm{m}$  under the private key scheme described in Section \ref{HSM:sec4} and the uniform distribution on $\mathbb{Z}_q^\ell$
\end{lemma}
\begin{proof}
For $0\leq i \leq d+\ell-n$, let $\bm{P}_i$ be the matrix got by taking the first $d+i$ rows of $pk$ ($\bm{P}_0 = \bm{C}_0$). Let $\mathcal{D}_i$ denote the set of vectors got by taking the sum of subsets of the rows of $\bm{P}_i$ i.e., $\mathcal{D}_i := \{\bm{x} \bm{P}_i:\bm{x} \in \{0,1\}^{d+i}\}$. Note that, for $j \leq i$, $\mathcal{D}_j \subset \mathcal{D}_i$. We claim that, for $0\leq i \leq d$, there exists no efficient algorithm that can distinguish between vectors sampled from a uniform distribution on $\mathcal{D}_i$ and vectors sampled from a uniform distribution on $\mathbb{Z}_q^\ell$. We prove this claim using induction. 

For a set $S$ and an algorithm $W$, let $p_W(S)$ be the probability that $W$ returns 1 when the input is sampled from a uniform distribution on $S$.
As a consequence of Lemma \ref{leftoverhash}, for any algorithm $W$ that takes as input elements of $\mathbb{Z}_q^{\ell}$,
\begin{equation}
    p_{W}(\mathcal{D}_0) - p_{W}(\mathbb{Z}_q^{\ell}) \leq 2^{-\omega (n)}
\end{equation}

 Assume that the claim is true for $i \leq k$. The $k+1$-th row of $\bm{P}_{k+1}$ is a random encryption of the message $\bm{b}_{k+1}$. Let $\mathcal{D}_k^{\prime}$ be the following set of vectors,
 \begin{equation}
     \mathcal{D}_k^{\prime} := \{\bm{P}_{k+1}(:,k+1) + \bm{v}\,|\, \bm{v} \in \mathcal{D}_k \}
 \end{equation}
 Clearly $\mathcal{D}_{k+1} := \mathcal{D}_k \cup \mathcal{D}_k^{\prime} $. Therefore, for any algorithm $W$
 \begin{equation}
     p_W(\mathcal{D}_{k+1}) = \frac{1}{2}p_W(\mathcal{D}_k) + \frac{1}{2}p_W(\mathcal{D}_{k}^{\prime})
 \end{equation}
 Therefore,
 \begin{equation}\label{dk+1}
     p_W(\mathcal{D}_{k+1})- p_{W}(\mathbb{Z}_q^{\ell})= \frac{1}{2}((p_W(\mathcal{D}_k)-p_{W}(\mathbb{Z}_q^{\ell})) + (p_W(\mathcal{D}_{k}^{\prime})-p_{W}(\mathbb{Z}_q^{\ell})))
 \end{equation}
The term $(p_W(\mathcal{D}_k)-p_{W}(\mathbb{Z}_q^{\ell}))$ is negligible by the induction assumption. Therefore, if the LHS term in Equation \ref{dk+1} is non-negligible, then $(p_W(\mathcal{D}_{k}^{\prime})-p_{W}(\mathbb{Z}_q^{\ell}))$ must be non-negligible. Consider the distribution of vectors $\bm{v} = \bm{v}_1 +\bm{v}_2$ where $\bm{v}_2$ is sampled from a uniform distribution on $\mathcal{D}_{k}$. This distribution is uniform if $\bm{v}_1$ is sampled from the uniform distribution on $\mathbb{Z}_q^\ell$. On the other hand, if $\bm{v}_1$ is sampled uniformly at random from the set of encryptions of $\bm{b}_{k+1}$ then the distribution of $\bm{v}$ is uniform in $\mathcal{D}_{k}^{\prime}$. Thus, if $(p_W(\mathcal{D}_{k}^{\prime})-p_{W}(\mathbb{Z}_q^{\ell}))$ is non-negligible, then the algorithm $W$ can be used to distinguish between vectors that are sampled from the uniform distribution on the encryptions of $\bm{b}_{k+1}$ (under the private key scheme) and vectors that are sampled from the uniform distribution on $\mathbb{Z}_q^{\ell}$. \qed
\end{proof}

\section{Conclusion}\label{ch4:sec:sum}
In this paper, we have described a leveled fully homomorphic encryption scheme which achieves additive and multiplicative homomorphism without key switching, the security of which depends on the hardness of the LWE problem. This scheme can be used to encrypt message vectors and the addition and multiplication is done bit wise. The computation cost per multiplication is $\tilde{\mathcal{O}}(n^3 \cdot L^2)$.

\begin{acknowledgements}
The authors are grateful to Dr. Vinay Wagh, Department of Mathematics, Indian Institute of Technology Guwahati for his valuable suggestions.
\end{acknowledgements}

\bibliographystyle{spmpsci} 
\bibliography{CC_ref}

\end{document}